%% file: main.tex
\documentclass[11pt]{article}
\pdfoutput=1
\usepackage{amsfonts,amsmath,amssymb,amsthm,graphicx}
\usepackage{changepage}
\usepackage{enumerate}
\usepackage{accents}
\usepackage{hyperref}
\usepackage{xfrac}
\usepackage[fleqn]{mathtools}
\usepackage{setspace}
\usepackage{ifthen}

\usepackage{fullpage}
\usepackage[margin=1in]{geometry}

\usepackage{natbib}

\usepackage[capitalize]{cleveref}

\usepackage{tikz}
\usepackage{pgfplots}
\usetikzlibrary{patterns}
\usetikzlibrary{intersections}
\usepgfplotslibrary{fillbetween}

\usepackage{thm-restate}

\theoremstyle{plain}
\newtheorem{theorem}{Theorem}
\newtheorem{lemma}{Lemma}

\newtheorem{proposition}{Proposition}

\theoremstyle{definition}
\newtheorem{definition}{Definition}
\newtheorem{example}{Example}
\newtheorem{assumption}[definition]{Assumption} 
\newtheorem*{remark}{Remark}

\theoremstyle{plain}

\let\oldparagraph\paragraph
\renewcommand{\paragraph}[1]{\oldparagraph{#1.}}

\begin{document}
\include{notation}
\include{math}

\title{Revenue Maximization for Buyers with Costly Participation}
\author{Yannai A. Gonczarowski\thanks{Department of Economics and Department of Computer Science, Harvard University.
Email: \texttt{yannai@gonch.name}. Research carried out while at Microsoft Research New England Lab.} 
\and Nicole Immorlica\thanks{Microsoft Research New England Lab.
Email: \texttt{nicimm@gmail.com}. }
\and Yingkai Li\thanks{Cowles Foundation for Research in Economics, Yale University.
Email: \texttt{yingkai.li@yale.edu}.
Yingkai Li thanks Sloan Research Fellowship FG-2019-12378 for financial support. Research carried out while this author was an intern in Microsoft Research New England Lab and a PhD candidate at Northwestern University.
} 
\and Brendan Lucier\thanks{Microsoft Research New England Lab.
Email: \texttt{brlucier@microsoft.com}. }}
\date{}
\maketitle

\begin{abstract}
We study mechanisms for selling a single item when buyers have private costs for participating in the mechanism. An agent's participation cost can also be interpreted as an outside option value that she must forego to participate. This substantially changes the revenue maximization problem, which becomes non-convex in the presence of participation costs. For multiple buyers, we show how to construct a $(2+\epsilon)$-approximately revenue-optimal mechanism in polynomial time. Our approach makes use of a many-buyers-to-single-buyer reduction, and in the single-buyer case our mechanism improves to an FPTAS. We also bound the menu size and the sample complexity for the optimal single-buyer mechanism. Moreover, we show that posting a single price in the single-buyer case is in fact optimal under the assumption that either (1) the participation cost is independent of the value, and the value distribution has decreasing marginal revenue or monotone hazard rate; or (2) the participation cost is a concave function of the value. When there are multiple buyers, we show that sequential posted pricing guarantees a large fraction of the optimal revenue under similar conditions.


\end{abstract}

\input{content/intro}
\input{content/prelim}
\input{content/optimal}

\input{content/pricing}
\input{content/conclusion}

\bibliographystyle{apalike}
\bibliography{ref}

\newpage
\appendix
\input{content/appendix}

\end{document}

%% file: notation.tex
\newcommand{\val}{v}
\newcommand{\cost}{c}
\newcommand{\alloc}{x}
\newcommand{\pay}{p}
\newcommand{\util}{u}
\newcommand{\rev}{R}
\newcommand{\virtual}{\phi}
\newcommand{\supp}{{\rm supp}}
\newcommand{\quant}{q}
\newcommand{\rv}{R}
\newcommand{\concaveRv}{\bar{\rv}}
\newcommand{\virtualv}{\Phi}
\newcommand{\jointf}{\bar{f}}
\newcommand{\jointF}{\bar{F}}
\newcommand{\cl}{\mathcal{L}}
\newcommand{\hval}{\hat{\val}}
\newcommand{\cF}{\mathcal{F}}
\newcommand{\inftynorm}[1]{\left\lVert#1\right\rVert_{\infty}}
\newcommand{\opt}{{\rm OPT}}
\newcommand{\Rev}{{\rm Rev}}
\newcommand{\mech}{\mathcal{M}}
\newcommand{\bval}{\bar{\val}}
\newcommand{\bF}{\bar{F}}

\newcommand{\type}{t}
\newcommand{\BNE}{{\rm BNE}}
\newcommand{\primed}{''}
\newcommand{\vupper}{H}
\newcommand{\poly}[1]{{{\rm poly}\left(#1\right)}}
\newcommand{\order}{\tau}
\newcommand{\info}{\mathcal{E}}
\newcommand{\dprimed}{^{\dagger}}
\newcommand{\ddprimed}{^{\ddagger}}

\newcommand{\optd}{\widehat{\opt}}
\newcommand{\optt}{\widetilde{\opt}}
\newcommand{\grid}{\Psi}

%% file: math.tex
\newcommand{\setsize}[1]{{\left|#1\right|}}

\newcommand{\floor}[1]{
{\lfloor {#1} \rfloor}
}
\newcommand{\bigfloor}[1]{
{\left\lfloor {#1} \right\rfloor}
}

\newcommand{\super}[1]{^{(#1)}}

%
%
\newcommand{\given}{\,\mid\,}

\newcommand{\prob}[2][]{\text{\bf Pr}\ifthenelse{\not\equal{}{#1}}{_{#1}}{}\!\left[{\def\givenn{\middle|}#2}\right]}
\newcommand{\expect}[2][]{\text{\bf E}\ifthenelse{\not\equal{}{#1}}{_{#1}}{}\!\left[{\def\givenn{\middle|}#2}\right]}

\newcommand{\tparen}{\big}
\newcommand{\tprob}[2][]{\text{\bf Pr}\ifthenelse{\not\equal{}{#1}}{_{#1}}{}\tparen[{\def\given{\tparen|}#2}\tparen]}
\newcommand{\texpect}[2][]{\text{\bf E}\ifthenelse{\not\equal{}{#1}}{_{#1}}{}\tparen[{\def\given{\tparen|}#2}\tparen]}

\newcommand{\sprob}[2][]{\text{\bf Pr}\ifthenelse{\not\equal{}{#1}}{_{#1}}{}[#2]}
\newcommand{\sexpect}[2][]{\text{\bf E}\ifthenelse{\not\equal{}{#1}}{_{#1}}{}[#2]}

\newcommand{\rbr}[1]{\left(#1\right)}

\newcommand{\suchthat}{\,:\,}

\newcommand{\partialx}[2][]{{\tfrac{\partial #1}{\partial #2}}}
\newcommand{\nicepartialx}[2][]{{\nicefrac{\partial #1}{\partial #2}}}
\newcommand{\dd}{{\,\mathrm d}}
\newcommand{\ddx}[2][]{{\tfrac{\dd #1}{\dd #2}}}
\newcommand{\niceddx}[2][]{{\nicefrac{\dd #1}{\dd #2}}}
\newcommand{\grad}{\nabla}

\newcommand{\symdiff}{\triangle}
\newcommand{\abs}[1]{\left|#1\right|}
\newcommand{\indicate}[1]{{\bf 1}\left[#1\right]}

\newcommand{\reals}{\mathbb{R}}

%% file: content/intro.tex
\section{Introduction}

In the textbook example of revenue optimization, a monopolist seller wishes to sell a single item to 
a collection of buyers whose values are drawn independently from known distributions.
In this idealized economic environment, the nature of the optimal mechanism was solved by \citet{myerson1981optimal}.  
However, actual economic environments depart significantly from this ideal.  One significant assumption is that  
each buyer's cost from participation is normalized to zero.
While convenient, this assumption is not without loss: it requires that agents are perfectly indifferent between participating in an auction and leaving empty-handed, versus not attending the auction in the first place.

In this paper we relax this assumption, allowing each buyer to have a private cost for participation in the mechanism (or, equivalently, a private value for an outside option she must forgo to participate). 
Thus, a buyer's type includes both a valuation of the item for sale and a cost for participation, 
and the seller chooses a mechanism for selling the item to maximize the expected revenue given a Bayesian prior over each buyer's type. 
This is a generalization of the model in \citet{menezes2000auctions} and \citet{celik2009optimal} where the participation cost is public information. 

Our model can be interpreted in various ways. 
One interpretation is an oligopoly model with 
exclusive competition \citep{champsaur1989multiproduct}. 
For example, suppose there are multiple firms who wish to sell different items to a single buyer by holding auctions simultaneously in different locations, 
and the buyer can participate in at most one auction.
Then from the perspective of any one seller, if we fix the other sellers' strategies (selling mechanisms), the utility of the buyer for participating in a competing auction can be modeled as a private cost for participation.
Thus, the best response problem in this oligopoly model is 
the optimal mechanism for buyers with participation costs.
Another example is to view our model as a setting where buyers must make investment decisions for boosting their values for winning the auction \citep{gershkov2018theory}. Without investment, the buyers have no value for winning the auction.
The investments are costly for the buyers,
and the buyers must decide whether to invest before the auction starts.
In these interpretations, the previous papers \citep{champsaur1989multiproduct,gershkov2018theory} have focused on settings equivalent to public, commonly known participation costs. 
In contrast, the model that we study has private participation costs, which can be arbitrarily correlated to the value of the item.

At first glance, one may wonder whether this setting is in fact equivalent to a single-dimensional setting in which a buyer's type is simply her value minus her participation cost. 
However, such a transformation can substantially impact a buyer's choices.
Consider a single buyer with the choice between paying $\$3.1$ for getting the item for sure or paying $\$1$ for getting the item with probability $50\%$\footnote{Note that with the remaining $50\%$ probability, the buyer will not get the item and she still needs to pay the cost since she chose to participate in the mechanism.} (or not participating at all).
If the buyer has a value of $\$5$ for the item and a participation cost of $\$1$, she would strictly prefer the first option (paying $\$3.1$ and getting the item for sure, for a resulting utility of $5-3.1-1=\$0.9$).
However, if instead of having a value of $\$5$ and a participation cost of $\$1$, the buyer has a value of $5\!-\!1\!=\!\$4$ and a participation cost of $\$0$, she would strictly prefer the second option (paying $\$1$ and getting the item with probability $50\%$, for a resulting expected utility of $\frac{1}{2}\cdot4-1=\$1$). 
So there is more at play here than is conveyed by the difference between the value and the participation cost.

As we show, known results about the structure of the revenue-maximizing mechanism in the standard setting without participation costs may fail to hold in the presence of participation costs,
even in the single-buyer case.
Indeed, even when there is only one buyer,
the seller can strictly benefit by offering lotteries (see \cref{exp:lottery better} below). 
This is in sharp contrast to the idealized environment with participation costs normalized to zero, where it is always optimal to simply post a single take-it-or-leave-it price~\citep{myerson1981optimal}.  
We therefore embark on a study to adapt and extend canonical economic and computational results on revenue maximization to this generalized environment.

\subsection{Our Contributions}

\paragraph{A Characterization of Incentive Compatible Mechanisms}  We begin by providing a structural characterization of incentive compatible mechanisms in the presence of participation costs.
Recall that in the classic setting where participation costs are normalized to zero, a (direct-revelation) mechanism has each buyer $i$ declare her valuation $v_i$, then maps this profile of values into an allocation $x_i$ and payment $p_i$ for each buyer.  
In our setting, a direct-revelation mechanism takes as input a tuple $(v_i, c_i)$ of value and participation cost from each agent. 
The question is which (two-dimensional) allocation and payment rules are incentive compatible and individually rational in this environment.

For a single buyer, one way to construct a truthful mechanism is to ignore the participation cost entirely and choose a monotone allocation rule that depends only on the reported value.  The agent is then free to opt out whenever her expected utility under the allocation rule is less than her participation cost. We call such mechanisms \emph{opt-out-or-revelation mechanisms}.  
Our first result is that, in fact, \emph{every} mechanism for a single buyer is revenue-equivalent to an opt-out-or-revelation mechanism, which further allows us to therefore restrict attention to such mechanisms without loss, and incentive compatibility reduces to monotonicity of the allocation rule.  This enables us to establish a payment identity for our setting, analogous to the classic one of \citet{myerson1981optimal}, and establish a virtual-value interpretation of revenue maximization in the presence of participation costs.

We would like to extend this characterization of incentive compatibility to multiple buyers.  However, since each agent's expected utility from the mechanism can depend on which other agents participate, we must be careful when modeling participation choices.  Given a fixed allocation rule, there may be multiple equilibria of the ``participation game'' between agents choosing whether or not to opt in.\footnote{For example, consider the mechanism that chooses a participating agent at random to receive the item for free. If there are two agents, each with value $3$ and participation cost $2$, then there will be three equilibria of participation: one where only the first agent participates, one where only the second agent participates, and one where each agent participates with probability $2/3$.}  The choice of equilibrium can impact the mechanism's revenue, and hence what we mean by the revenue of a given allocation rule is potentially ambiguous.

One way to resolve this ambiguity is to encode participation choices into the mechanism protocol itself.  For example, we could imagine the mechanism approaching agents one at a time, asking them sequentially whether they would like to opt in or out (potentially after revealing information about the agents who had previously chosen to opt in).  After all participation decisions have been made, the mechanism then applies a direct-revelation allocation rule that depends only on the participating agents' reported values.  We call such mechanisms \emph{sequential opt-out-or-revelation mechanisms}.  As in the single-buyer case, we show that every mechanism (and every equilibrium of agent participation behavior) is outcome-equivalent to a sequential opt-out-or-revelation mechanism.  
This allows us to extend our characterization of incentive compatible mechanisms and our virtual-value interpretation of revenue maximization to the setting with multiple buyers.

\paragraph{Computing an Approximately Revenue-Optimal Mechanism}

Characterization in hand, we next turn to the problem of constructing an approximately revenue-optimal mechanism.  Our benchmark is the highest revenue achievable by any mechanism and any equilibrium of participation choices by the agents.  We show how to construct a mechanism that achieves a constant fraction of this benchmark, less an additive $\epsilon$ loss.  Moreover, the equilibrium of participation in our mechanism is unique up to tie-breaking.\footnote{In other words, an agent might be indifferent between participating or not participating, and either choice could be supported at equilibrium.  Such indifference occurs with vanishing probability and does not impact the mechanism's revenue.}

\begin{theorem}
\label{thm:intro.many.agent.mech}
For any $\epsilon > 0$ and any $n$ buyers with types drawn from a product distribution supported on $[0,1]^{2n}$, a mechanism with expected revenue $\frac{1}{2}\opt-\epsilon$ can be computed in time
polynomial in $n$, $1/\epsilon$, 
and either the maximum density of the participation cost in any agent's type distribution (for continuous distributions) or the maximum support size of any one agent's type distribution (for discrete distributions).
\end{theorem}

Our assumption only requires independence across different buyers. 
For each buyer, the value can be arbitrarily correlated with the participation cost.
Our construction employs a reduction from the many-buyer mechanism design problem to a single-agent problem.  Such methods have been used with great success in related settings, such as selling to agents with budget constraints and other revenue optimization problems~\citep{alaei2014bayesian,alaei2012bayesian}.  This reduction framework typically has three steps:
\begin{enumerate}
    \item Compute an approximately revenue-optimal mechanism for a single buyer, possibly subject to additional constraints (e.g., an upper bound on the ex-ante probability of selling the item).
    \item Optimize over the choice of constraints for each buyer participating in the mechanism, then construct the optimal interim allocation rule for each buyer subject to their constraints.
    \item Combine the constructed interim allocation rules into a multi-agent ex post allocation rule.
\end{enumerate}

As it turns out, implementing each of these three steps poses substantial challenges in our setting with participation costs.  We will discuss each of them in turn.

\bigskip \noindent
\textit{Step 1: Solving The Single-Buyer Problem.}  

Our first step is to construct a revenue-optimal mechanism for a single buyer.  Recall that when there is only one buyer, we can restrict attention to opt-out-or-revelation mechanisms.  For this case we provide an FPTAS algorithm for computing the optimal mechanism.

\begin{theorem}
\label{thm:intro.single.agent.mech}
For any $\epsilon > 0$ and any distribution over buyer types supported on $[0,1]^2$, a single-buyer mechanism with expected revenue $\opt-\epsilon$ can be computed in time $\poly{1/\epsilon}$.
\end{theorem}

A key challenge in this mechanism design problem is that revenue maximization is inherently non-convex in the presence of participation costs.  Recall that a given buyer type will opt out of the mechanism entirely if her expected utility drops below her participation cost.  This means that over the space of incentive compatible allocation rules, revenue can be a discontinuous and not necessarily concave function of allocation probability.\footnote{For example, suppose the agent's type $(v_i, c_i)$ is equally likely to be $(7,4)$ or $(3,1)$.  A mechanism that sells the item for sure at a price of $2.5$ would sell to the agent of type $(7,4)$, whereas the agent of type $(3,1)$ wouldn't participate.  On the other hand, a mechanism that charges a price of $1.2$ for a $50\%$ chance at the item would sell to the agent of type $(3,1)$ whereas agent type $(7,4)$ would opt out.  But a mechanism that randomizes uniformly between these two allocation rules obtains revenue $0$, since neither agent type would choose to participate.}  This immediately rules out many convex programming and duality-based approaches to revenue maximization that are common in the algorithmic mechanism design literature.

To address this challenge, we instead discretize the type space and directly optimize over feasible allocation rules.  This involves tracking not only the allocation rule itself, but also the utility obtained by each buyer type, as this is necessary for determining which buyer types will opt into the mechanism.

We note that the mechanism returned by our FPTAS can have a menu of lotteries of size potentially linear in the number of buyer costs.  We show that this is not an artefact of the approximation:
the menu size of the revenue-maximizing mechanism is also at most linear in the number of possible costs.
Moreover, this bound on the menu size is tight up to a multiplicative factor of $2$, even in the special case that the participation cost is perfectly correlated with the value.
We also bound the number of samples required to learn an up-to-$\epsilon$ revenue-maximizing mechanism absent direct access to the underlying distribution.

\bigskip \noindent
\textit{Step 2: Optimizing Bounds on the ex-ante Allocation Probabilities}

We now turn back to the multi-agent mechanism design problem.
In any valid mechanism, the total sum of ex-ante probabilities of allocating the item to each buyer can be at most $1$.  The next step of our reduction is to choose how to divide this unit of probability among the agents in our mechanism.  Unfortunately, we face the same challenge as in the single-buyer problem: the non-convexity of revenue maximization in the presence of participation costs.
This blocks us from using convex-programming techniques to optimize over the allocation constraints, as is typical for applications of this approach.

We address this challenge by discretizing the set of potential ex-ante allocation probabilities, then directly optimizing over potential assignments of constraints to individual buyers via dynamic programming.  One important observation is that we must take a buyers' participation decisions into account when estimating the ex-ante probability of allocation.  I.e., if we want to find the allocation rule that optimizes revenue subject to allocating the item with total probability at most~$q$, then any buyer type who opts out of the mechanism should not count toward this $q$, and of course the allocation rule itself determines which types opt in or out.  Moreover, these decisions can be distorted by the discretization of the type space.  For this reason, we permit some slack in our ex-ante allocation constraints and we must bound the accumulation of errors in the estimated probability of allocation.

{An additional complication arises when adapting our single-buyer FPTAS analysis to revenue maximization with ex-ante allocation constraints. 
As it turns out, the revenue-optimal mechanism that sells with probability at most $q < 1$ might not be a single allocation rule.  Rather, the seller may want to randomize over multiple mechanisms with different (and potentially also random) allocation and payment rules.  This can be beneficial because the seller is permitted to announce the realized mechanism before the buyer chooses whether to opt in, and doing so could influence the buyer's participation decision.
Searching over all possible distributions over mechanisms is computationally infeasible. 
Fortunately, we show that to maximize revenue subject to an ex-ante allocation probability it suffices to consider distributions over at most two mechanisms, each satisfying one of our discretized allocation probabilities.
The optimal mechanism in this restricted class can be computed efficiently by dynamic programming. 
}

\bigskip \noindent
\textit{Step 3: Contention Resolution.}

The final step in our construction is to combine the single-agent interim allocation rules into a single multi-agent allocation mechanism.  One way to do this is to try allocating to each agent independently using her own interim allocation rule, then use contention resolution techniques if any conflicts arise. (I.e., if we try to allocate to multiple agents at the same time, choose which of them should receive the item, if any.)  Unfortunately, the presence of participation costs once again creates a problem: contention resolution modifies the interim allocation rule experienced by each agent, and this can influence each agent's decision of whether to participate in the mechanism. 

We address this challenge by employing an \emph{online} contention resolution scheme, such that the order in which we resolve allocations is aligned with a sequential opt-out-or-revelation implementation of our mechanism.  Specifically, we employ a variation of an online rounding method due to~\cite{alaei2012online}.  For each agent sequentially, the mechanism will pre-determine whether that agent is eligible to participate in the mechanism or not.  If so, that agent will face an allocation rule that is identical to the single-agent interim rule constructed in our reduction, and hence her participation incentives will be unchanged.  An appropriate choice of eligibility probabilities for each agent yields an ex post implementable allocation rule while reducing the total revenue by at most one half.  This ultimately leads to the mechanism promised in Theorem~\ref{thm:intro.many.agent.mech}.

\paragraph{Conditions for the Optimality of Posted-Price Mechanisms}

Even for a single buyer, we've shown that the optimal mechanism may require the use of lotteries. But are there conditions under which it is optimal to post a fixed take-it-or-leave-it price, as is the case without participation costs?

We show that if 
either (1) the valuation distribution is independent of the participation cost distribution and has decreasing marginal revenue or monotone hazard rate (\cref{asp:mhr});
or (2) the participation cost is 
a concave function of the value;
then the optimal mechanism simply posts a (carefully chosen) single price for the item. 
It turns out that, under these conditions, the optimal price is in fact equal to the standard monopoly price for the (single-dimensional) distribution of the difference between the value and the participation cost.
This is not a coincidence: conditional on only posting a single take-it-or-leave-it price, a mechanism cannot distinguish between types with the same difference between value and participation cost, and hence the mechanism's revenue depends only on this difference.

Finally, in the multi-buyer setting, 
as suggested by \citet{gershkov2018theory}, 
characterizing the exact revenue-optimal mechanism may not be tractable 
even when the participation costs are known 
and all buyers are symmetric. 
This is because the buyers do not satisfy the von Neumann–Morgenstern expected utility characterization
for their preference over lotteries,
and the revenue-optimal mechanism may not be symmetric even in this simplified case, as the set of feasible mechanisms is not convex (see Appendix \ref{apx:multi-buyer} for more details).
In contrast, we show that even in the asymmetric setting, 
if for each buyer,
either (1) the valuation distribution is independent of the participation cost distribution and has decreasing marginal revenue;
or (2) the participation cost is 
a concave function of the value;
then a sequential posted price mechanism guarantees a constant fraction of the optimal revenue.\footnote{Our result for the multi-buyer setting allows for a setting in which some of the buyers satisfy condition (1) while the others satisfy condition (2).}  Relative to our general mechanism construction from Theorem~\ref{thm:intro.many.agent.mech}, this result imposes more constraints on the buyers but yields an improved approximation factor and takes the simpler form of sequential take-it-or-leave-it prices.

\subsection{Related Work}
\label{sec:relatedwork}

Our paper closely relates to the literature on auctions with private outside options and endogenous participation. 
\citet{rochet2002nonlinear} illustrate that some general lessons from that work are not robust to the presence of a private value for the outside option when there are production costs (which our model does not have).
Follow-up work in economics
focuses mainly on qualitative features of the optimal mechanisms with costly participations,
such as showing that distortion at the top or bottom is not required. 
There are various extensions of the model by considering the optimal mechanisms when selling congestible goods \citep{jebsi2006optimal}, 
when principals are risk-averse \citep{basov2010optimal}, 
or when considering the optimal income taxation rule \citep{lehmann2011optimal}.
A recent paper by \citet{ashlagi2021costly} considers consumer surplus maximization when buyers have outside options, 
and buyers can only misreport their outside-option values by downward deviation. 
Relative to that economics literature, 
our work provides a general characterization of the optimal mechanism through a payment identity and virtual-value analysis,
polynomial time algorithms for computing the (approximately) optimal mechanisms,
and natural sufficient distributional conditions for price posting to be optimal or to guarantee a large fraction of the optimal revenue.

There is a body of literature examining mechanism design with perfectly correlated outside options, i.e., the outside-option value is a deterministic and publicly known function of each agent's item value \citep{jehiel1996not,krishna1998efficient,jullien2000participation,figueroa2009role}.\footnote{\citet{jehiel1996not,figueroa2009role} motivate the perfectly correlated outside options through agents' externalities in allocations. In their model, the principal can control the outside option each agent receives, 
while in our model, the outside option values are exogenous and stochastic. } 
The assumption of perfectly correlated outside options significantly simplifies the optimal design problem, as illustrated in \citet{jullien2000participation} and \cref{sec:correlated} of our paper. 
However, in our study, the main challenges arise when the agents' outside option values are still stochastic even conditional on their item values. 

The model of costly participation also relates to the literature on return on investment (ROI) constraints \citep{golrezaei2021bidding,golrezaei2021auction,lucier2023autobidders}. 
In these papers, the ROI constraints for each agent can be seen as the additional utility the agent can gain by investing money in activities beyond the auction.

Our setting also relates to so-called ``one-and-a-half dimensional'' or ``interdimensional'' mechanism design settings.  In that domain, a buyer has one dimension representing her willingness to pay, 
and an additional ``half dimension'' representing a constraint
on her demand. 
For example, in the FedEx problem \citep{fiat2016fedex} 
the ``half dimension'' is the buyer's deadline,
and the buyer is only willing to accept the item before this deadline. 
For buyers with budgets \citep{che1998standard,devanur2017optimal}, 
the ``half dimension'' is the buyer's budget constraint, which places a cap on the maximum value she can pay for the item offered by the seller. 
In our setting, 
a buyer's participation cost can also be viewed as a ``half dimension'' of sorts 
since it only affects her decision of whether or not to participate in the auction;
conditional upon participating in the auction,
the cost does not affect the buyer's utility in the auction.  

Our paper also falls into the scope of designing
simple and approximately optimal mechanisms \citep{hartline2009simple}. 
This line of work mainly focuses on the case in which all buyers have zero cost for participation. 
For the single-item setting, sequential posted pricing guarantees a $1\!-\!\sfrac{1}{e}$ fraction of the optimal revenue \citep{yan-11}. 
For multi-item setting, 
sequential posted pricing guarantees a constant fraction of the optimal revenue when buyers have unit-demand valuations \citep{CHMS-10,cai2019duality},
and selling items separately or as bundles guarantees a constant fraction of the optimal revenue when buyers have additive valuations \citep{babaioff2020simple,cai2019duality}.
When buyers have non-linear utilities (e.g., budgeted utility), 
\citet{feng2020simple} show that constant-fraction approximation results (e.g., sequential posted pricing) for linear buyers 
can be generalized to non-linear buyers 
when the type distributions of the buyers satisfy some closeness property. 
There are numerous results on this topic. 
See the survey of \citet{roughgarden2019approximately} for a detailed discussion
on approximately optimal mechanisms in various other settings.

Finally, our results on the menu and sampling complexity of the revenue-optimal mechanism also contribute to the literature on menu sizes of optimal and approximately optimal mechanisms \citep[e.g.,][]{fiat2016fedex,hart2017approximate,babaioff2017menusize,devanur2017optimal,gonczarowski2018bounding,saxena2018menu,devanur2020optimal} and to the literature on the sample complexity of learning up-to-$\epsilon$ optimal mechanisms \citep[e.g.,][]{cole2014sample,morgenstern2015pseudo,devanur2016sample,gonczarowski2017efficient,hartline2019sample,gonczarowski2018sample,guo2019settling}.

\subsection{Roadmap}

In Section~\ref{sec:model} we formalize our model and describe our revelation principle: that any mechanism is revenue-equivalent to a truthful equilibrium of a sequential opt-out-or-revelation mechanism.  In Section~\ref{sec:opt} we characterize the truthful allocation rules for sequential opt-out-or-revelation mechanisms, and provide a virtual-value interpretation of revenue maximization.

In Section~\ref{sec:single} we consider revenue optimization for a single buyer.  We present our FPTAS in Section~\ref{sec:fptas}, and bound the menu and sample complexity of the optimal mechanism in Sections~\ref{sec:menu complexity} and~\ref{sec:discrete continuous}. Continuing with the single-buyer setting, in Section~\ref{sec:pricing} we present sufficient conditions for a posted-price mechanism to be revenue optimal.

In Section~\ref{sec:multi} we turn to the multi-agent setting.  We present our $O(1)$-approximate mechanism in Section~\ref{sec:multi-const}.  In Section~\ref{sec:multi-pricing} we establish conditions under which sequential posted pricing is approximately optimal.  We discussion two alternative models for costly participation and open questions in Section~\ref{sec:conclude}.

%% file: content/prelim.tex
\section{Model and Preliminaries}
\label{sec:model}
A seller has a single item to offer for sale to $n$ buyers. 
Each buyer~$i$ has a private valuation $\val_i\geq 0$
for getting the item 
and a private cost $\cost_i$ for participating in the mechanism.\footnote{Although it is most natural to consider non-negative values for participation costs, 
negative values for participation costs can capture the potential social benefit for the buyer to participate the auction,  
which is not modeled in her valuation for the item.
Introducing negative values for participation costs will lead to different observations for the optimality of posted pricing, 
which will be discussed formally in \cref{apx:negative outside}.
}
Let $\type_i \equiv (\val_i,\cost_i)$ represent the private type of buyer $i$
and let $\jointF_i$ be the seller's prior distribution over $\type_i$ with density $\jointf_i$.
Let $F_i$ and $G_i$ be the marginal distributions of values and costs
with densities $f_i,g_i$.
Note that a buyer's value and cost may be correlated under $\jointF_i$.
Let $\jointF = \times_i \jointF_i$ be the product distribution of the buyers' type profiles.

\paragraph{Mechanisms}

A (possibly non-direct-revelation) mechanism is a communication protocol. This protocol can be represented by a game tree that determines which players can (simultaneously) send and/or receive messages at each round (i.e., node of the tree) and how those messages influence the progression of the protocol.  The leaves of the game tree specify rounds in which the protocol halts, at which point the mechanism terminates and outputs the allocation and payments for each buyer.  To capture participation decisions, we will assume that each agent's message space includes a special message $\psi$.  If the first message (and only the first message) sent by an agent is $\psi$, then the agent is said to have opted out of the mechanism and they receive utility equals zero.


We provide a generalized notion of a revelation mechanism, such that 
any mechanism can be transformed into an outcome-equivalent \emph{sequential opt-out-or-revelation mechanism}.  Intuitively, in such a mechanism agents only report their values, and there are fixed allocation and payment rules that map these reports to the mechanism's outcome.  But the mechanism is not restricted to simultaneous reports: it can approach agents sequentially to solicit their reports, and when it is an agent's turn to report she can opt out of the mechanism rather than report her value.  We formalize this class of mechanisms below.

\begin{definition}
A \emph{sequential opt-out-or-revelation mechanism} proceeds as follows:
\begin{itemize}
    \item There is a (possibly random, possibly adaptive) order $\order$ over the agents.  For each agent $i$ in this order, the mechanism sends a message $\info_i$ that can depend on the protocol history.  Agent $i$ then either opts out (by reporting $\psi$) or reports a declared value $\tilde{v}_i \geq 0$.  Write $\bar{\reals} = \reals \cup \{\psi\}$ for this message space.
    
    \item Once all agents have reported, the outcome is determined by an allocation function $\alloc:\bar{\reals}^n\to \Delta(\{0,1\}^n)$ and payment function $\pay: \bar{\reals}^n \to \reals^n$, such that each buyer $i$ reporting $\psi$ receives $x_i=p_i=0$.

    \item There exists an equilibrium in which each buyer $i$ that opts into the mechanism reports her value $v_i$ truthfully. 
    That is,  
    there exists a profile of strategies $\sigma$
    such that for any buyer $i$, any $\sigma_i(t_i, \info_i)  \in \{\psi, v_i\}$, 
    and any order $\order$,
    \begin{align*}
    \expect[t_{-\order_i}\sim \jointF_{-\order_i}]{u_{\order_i}(\mech(\sigma(t, \info))) \given \info_{\order_i}} 
    \geq \expect[t_{-\order_i}\sim \jointF_{-\order_i}]{u_{\order_i}(\mech(b,\sigma_{-\order_i}(t_{-\order_i}, \info_{-\order_i})))\given \info_{\order_i}},\,
    \forall b\in \bar{\reals}.
    \end{align*}
\end{itemize}
\end{definition}

Note that this class includes traditional simultaneous-move direct revelation mechanisms, as this corresponds to each $\info_i$ being an empty message.  When there is only a single buyer, a sequential opt-out-or-revelation mechanism has an especially natural form: it is simply a direct-revelation mechanism (characterized by a truthful allocation and payment rule) in which each buyer type $(v_i, c_i)$ is permitted to opt out, and will do so precisely if her expected utility from the mechanism is less than $c_i$.

In \cref{sec:revelation} we prove that any mechanism protocol can be converted into a sequential opt-out-or-revelation mechanism with the same distribution over outcomes (and hence the same expected revenue). 

\begin{lemma}\label{thm:revelation}
For any type distribution $\jointF$, any mechanism $\mech$, and any Bayes-Nash equilibrium of the corresponding game, there exists a sequential opt-out-or-revelation mechanism $\widehat{\mech}$ and a truthful equilibrium of $\widehat{\mech}$ that generates the same distribution over outcomes.
\end{lemma}

\paragraph{Posted-Price Mechanisms}
When there is a single buyer ($n=1$), we say that a sequential opt-out-or-revelation mechanism $\mech$ \emph{posts a price} if there exists a price $p$ such that the mechanism awards the item with probability~$1$ for a price of $p$ if the buyer participates and reports a value of no less than $p$, and otherwise the item is not awarded and $p=0$. With $n\ge2$ buyers, we say that a sequential opt-out-or-revelation mechanism $\mech$ is a \emph{sequential posted price mechanism} if there exist prices~$p_i$ and an order over the buyers such that the mechanism awards the item with probability~$1$ for a price of $p_i$ to the first buyer $i$---according to this order---who participates and reports a value of no less than $p_i$ (and no other buyers are charged any price), and the item is not awarded and no buyer is charged any price if no such buyer exists. Each buyer is informed about the availability of the item when he sees the price.

\paragraph{Revenue Maximization}\sloppy
For any mechanism $\mech$, 
let $\Rev(\type,\sigma;\mech) = \sum_i p_i(\sigma(\type))$ 
be the revenue of mechanism $\mech$ given type profile $\type$ and strategy profile $\sigma$, 
and let $\Rev(\jointF, \sigma; \mech) = \expect[\type\sim\jointF]{\Rev(\type, \sigma;\mech)}$
be the expected revenue given distribution $\jointF$.
We will omit the strategy~$\sigma$ in the notation if it is clear from the context. 
Let 
\begin{align*}
\opt(\jointF) = \max_{\mech}\max_{\sigma\in \BNE(\jointF,\mech)}\Rev(\jointF,\sigma;\mech)
\end{align*}
be the optimal expected revenue of the seller given distribution $\jointF$
under the Bayes--Nash equilibrium with highest expected revenue.
Let $\phi(\val) = v-\frac{1-F(v)}{f(v)}$
be the virtual value of the buyer with value $v$ and valuation distribution~$F$.

\begin{definition}\label{def:regular}\label{asp:mhr}\label{asp:dmr}
For any valuation distribution $F$ with density function $f$ and virtual value function~$\virtual$, 
\begin{itemize}
\item $F$ is regular if the virtual value $\phi(\val)$ is non-decreasing;
\item $F$ has decreasing marginal revenue (DMR) if $f(\val)\phi(\val)$ is non-decreasing;
\item $F$ has monotone hazard rate (MHR) if $\phi'(\val)\geq 1$ for all $v$.
\end{itemize}
\end{definition}

%% file: content/optimal.tex
\section{Virtual Values and the Suboptimality of Posted Prices}
\label{sec:opt}

\cref{thm:revelation} shows that it is without loss to restrict to mechanisms that employ single-dimensional allocation and payment rules and give agents the opportunity to opt out.  But which allocation and payment rules are truthful for the buyers who choose to opt in?  In this section we show that, similar to the case without participation costs, truthfulness corresponds to monotonicity of interim allocation rules.  We establish a Myerson-style payment identity and virtual value characterization of truthful revenue maximization.  

For most of this section we will focus on the setting where there is only a single buyer in the market, so we will omit the subscript $i$ from our notation.
That said, all of the characterization results in this section extend directly to the multi-buyer setting by interpreting them as conditions on the interim allocation and payment experienced by each bidder; see the remark at the end of this section.
Note also that all results in this section apply even when the distribution over types allows for correlation between valuations and participation costs.  

In the following example, we show that
the private participation costs setting is qualitatively different from the setting in which the participation costs are normalized to zero,
by providing a distribution over values and participation costs---these will even be independently distributed---such that
offering lotteries to the buyer can strictly improve the expected revenue compared to posting a single fixed price.%
\footnote{There is an interesting conceptual connection between this example and the observation of \citet{deneckere1996damaged} that a seller may be able to strictly increase her revenue by introducing a damaged good into the market.
Note that our setting with a private participation cost can be converted into a two-item setting. In this two-item setting, one of the items would correspond to participating and winning while the other would correspond to participating and not winning---the latter having negative value with probability~$1$ (if the original participation cost is positive)---and the mechanism is constrained to ex-post allocate exactly one item to the buyer.
\cref{exp:lottery better} shows that price posting in our costly participation setting---
which in the transformed two-item setting translates to preventing the item with negative value from being sold---may not maximize revenue.
Therefore, in the transformed two-item setting, \cref{exp:lottery better} is interpreted as showing that the seller can strictly increase her revenue by introducing a new item with negative value---a good so damaged that its value to the buyer is in fact always negative.}
\begin{example}\label{exp:lottery better}
The value and cost of the buyer will be independently distributed.
The value distribution has CDF $F(\val) = 1 - \frac{1}{\val - 1}$ for $\val \in [2, 5)$ 
and $F(\val) = 1$ for $\val \geq 5$.
Note that $F$ is a regular distribution. 
The participation cost is $0$ with probability $0.15$ and~$2$ with probability $0.85$. 
In this example, the optimal posted-price mechanism is to post price $2$, resulting in revenue $0.8\bar{6}$. 
However, consider the mechanism that offers two probability-price lotteries $(1, 3), (\frac{2}{3}, \frac{4}{3})$. 
The buyer with value $5$ will choose the lottery $(1,3)$
and receives utility $2$ regardless of his participation cost. 
Moreover, the buyer with value in $[2,5)$ 
will choose the lottery $(\frac{2}{3}, \frac{4}{3})$ if his participation cost is $0$, 
and not participate in the auction if his participation cost is $2$. 
The expected revenue of this mechanism is $\frac{1}{4}\times 3 + \frac{3}{4} \times0.15 \times\frac{4}{3} = 0.9$, 
which is strictly higher than $0.8\bar{6}$.
The multiplicative gap between the optimal mechanism and optimal pricing is therefore higher than $1.038$.
\end{example}

A classic approach to designing the revenue-optimal mechanism for a single-item setting when the participation cost is known to the seller is to use a payment identity and virtual value analysis.

\begin{lemma}[\citealp{myerson1981optimal}]\label{thm:myerson}
In the single-item single-buyer setting, if the buyer has participation cost~$0$, 
for any truthful mechanism $\mech$ with allocation $\alloc$ and payment $\pay$,
it holds that $\alloc$ is non-decreasing and
the payment satisfies $\pay(\val) = \val \alloc(\val) - \int_0^{\val} \alloc(z) \,\dd z + \pay_0$ for some constant~$\pay_0$.
The expected revenue given distribution $F$ is 
$$\Rev(F;\mech)
= \expect[\val\sim F]{\alloc(\val)\phi(\val)} + p_0.$$
\end{lemma}

In this section, we provide an analog of the payment identity 
when buyers have private participation costs.  That is, we 
show that it is without loss to consider a truthful mechanism as if the participation cost is $0$,
in which case the buyer will participate if and only if her utility from participating is at least her participation cost.

\begin{lemma}\label{lem:revelation}
It is without loss for the seller to commit to a monotone non-decreasing allocation rule 
$\alloc(\val)$ and payment rule 
$\pay(\val) = \val \alloc(\val) - \int_0^{\val} \alloc(z) \dd z + \pay_0$ for some constant $\pay_0$, 
and for the buyer to participate and truthfully reveal $\val$ if and only if $\util(\val; \alloc, \pay) = \int_0^{\val} \alloc(z) \dd z - \pay_0 \geq \cost$
where $\cost$ is the participation cost. 
\end{lemma}
\begin{proof}
In \Cref{sec:model}, we have shown that 
it is without loss to consider the opt-out-or-revelation mechanism 
where the buyer truthfully reveals her value conditional on participation. 
By \cref{thm:myerson}, 
the allocation satisfying the incentive constraint must be non-decreasing,
and the payment satisfies 
$\pay(\val) = \val \alloc(\val) - \int_0^{\val} \alloc(z) \dd z + \pay_0$ for some constant $\pay_0$.
Finally, the utility of the buyer for participating the auction 
is $\util(\val; \alloc,\pay)
= \val \cdot \alloc(\val) - \pay(\val)= \int_0^{\val} \alloc(z) \dd z - \pay_0$,
and she therefore maximizes utility by participating in the auction if and only if her utility is at least $\cost$. 
\end{proof}

\medskip\noindent
Given any allocation rule $\alloc$ and corresponding payment $\pay(\val) = \val \cdot\alloc(\val) - \int_0^{\val} \alloc(z) \,\dd z + \pay_0$, 
let $\val_{\alloc}(\cost) = \inf_{\val\geq 0} \{\val \cdot\alloc(\val) - \pay(\val)\geq c\}$,
and 
\begin{align*}
\alloc_{\cost}(\val) = \begin{cases}
\alloc(\val) & \val \geq \val_{\alloc}(\cost)\\
0 & \val < \val_{\alloc}(\cost).
\end{cases}
\end{align*}
That is, $\val_{\alloc}(\cost)$ is the minimum value at which a buyer with participation cost $\cost$ will choose to participate in a mechanism with allocation rule $x$, and $x_\cost(\val)$ is the resulting allocation rule taking the participation decision into account.
For any joint distribution $\jointF$, 
let $\jointF_c$ be the conditional value distribution 
when the participation cost is $c$. 
We define the virtual value of the buyer given conditional valuation distribution $\jointF_c$
as $\virtual_c(\val) = \val - \frac{1-\jointF_c(\val)}{\hat{f}_c(\val)}$.

\begin{figure}[t]
\hspace{-5pt}
\begin{minipage}[t]{0.49\textwidth}
\centering
\setlength{\unitlength}{1cm}
\thinlines
\input{fig/payment_identity}
\end{minipage}
\begin{minipage}[t]{0.49\textwidth}
\input{fig/payment_identity_c}
\end{minipage}
\vspace{-8pt}
\caption{\label{fig:payment identity} The figure on the left is an illustration of 
the payment $\pay(\val)$ and the cutoff value $\val_{\alloc}(\cost)$
where $c-p_0\geq 0$. 
The figure on the right illustrates the interim allocation rule $\alloc_{\cost}(\val)$
when the participation cost is $\cost$. 
Here $\pay_{\cost}(\val)$ is the incentive compatible payment rule for allocation rule $\alloc_{\cost}(\val)$.
The figure illustrates that 
$\pay_{\cost}(\val) = \pay(\val) + c$
for any $\val\geq \val_{\alloc}(\cost)$.}
\end{figure}

\begin{restatable}{lemma}{lemRevEquiv}\label{thm:revenue equivalence}
Given any distribution $\jointF$ with marginal cost distribution $G$, 
and any mechanism $\mech$ with allocation $\alloc$
and payment rule $\pay$ with parameter $p_0$, 
the revenue of the seller is 
\begin{align*}
\Rev(\jointF;\mech) &= 
\expect[\cost\sim G]{\expect[\val\sim \jointF_c]
{\alloc_{\cost}(\val)\virtual_c(\val)}
- (1-\jointF_c(\val_{\alloc}(\cost)))\cdot 
\max\{-p_0, \cost\}}.
\end{align*}
\end{restatable}
The proof of \cref{thm:revenue equivalence} is given in \cref{sec:proof complexity}. 
Intuitively, the term $\expect[\val\sim \jointF_c]
{\alloc_{\cost}(\val)\virtual_c(\val)}$ represents the difference in social welfare and the agent's expected utility,
similar to the characterization in \citet{myerson1981optimal}.
The additional term $(1-\jointF_c(\val_{\alloc}(\cost)))\cdot 
\max\{-p_0, \cost\}$ in our setting is the additional utility the agent 
obtains from either saving her participation cost $c$, 
or the utility $-p_0$ the mechanism provides to the lowest type. 
Such additional utility for the buyer are subtracted to correctly calculate the expected revenue. 

In this paper, we focus on the problem of revenue maximization, 
so it is without loss to consider $p_0 \geq 0$.
When participation costs are non-negative, it is without loss to further assume that $p_0 = 0$.
The proof of the following lemma is given in \cref{sec:proof complexity}.
\begin{restatable}{lemma}{lempzz}\label{lem:p0=0}
For any mechanism with allocation and payment $\alloc, \pay$
such that $\pay(\val) = \val \alloc(\val) - \int_0^{\val} \alloc(z) \dd z + \pay_0$ for some constant $\pay_0$, 
there exists another mechanism with allocation and payment $\hat{\alloc}, \hat{\pay}$ 
with weakly higher revenue
and $\hat{\pay}(\val) = \val \hat{\alloc}(\val) - \int_0^{\val} \hat{\alloc}(z)\dd z + \hat{p}_0$
where $\hat{p}_0\geq 0$.
This can be strengthened to $\hat{p}_0 = 0$ if the participation costs are non-negative. 
\end{restatable}

\paragraph{Remark: Multiple Buyers}  We note that Lemmas~\ref{lem:revelation},~\ref{thm:revenue equivalence}, and~\ref{lem:p0=0} extend immediately to the case of multiple buyers by interpreting $x$ and $p$ as a buyer's interim allocation and payment rules.  In the multi-buyer case it is important to note that the interim allocation rule for agent $i$ is evaluated in expectation not only over the types of the other agents, but also over their (possibly randomized) participation decisions.  For weakly monotone interim allocation rules, the revenue obtained by the seller from each buyer is the virtual surplus less the participation cost adjustment, as in Lemma~\ref{thm:revenue equivalence}.

\section{Revenue Optimization for a Single Buyer}
\label{sec:single}

\subsection{An FPTAS Algorithm}
\label{sec:fptas}
We now turn to the problem of optimizing over the space of allocation rules identified in the previous section.
We will assume that both the value and the cost are supported in $[0,1]$.\footnote{Any bounded distribution can be normalized such that the support is in $[0,1]$. 
} 
We show that by discretizing the valuation space and the allocation space of the mechanism, 
the loss in optimal revenue is small, 
and for the discretized problem, the optimal mechanism can be computed efficiently using dynamic programming. 
The following two lemmas quantify the discretization errors, 
with proofs provided in \cref{sec:proof complexity}.

\begin{restatable}{lemma}{thmdiscreteerror}\label{thm:discrete}
Let $(\Omega, \cF, P)$ be any probability measure, 
and let $t_1, t_2 : \Omega \to \reals^2$ be two $2$-dimensional random variables. 
If $\sup_{\omega\in\Omega} 
\inftynorm{t_1(\omega) - t_2(\omega)} \leq\epsilon$, 
then $|\opt(t_1) - \opt(t_2)| \leq 3\epsilon$, 
where $\opt(t_k)$ is the optimal expected revenue when the valuation and participation cost follow the same distribution as the random variable $t_k$.
\end{restatable}
\vspace{-8pt}
\begin{restatable}{lemma}{thmdiscretealloc}\label{thm:discrete alloc}
For any distribution $\jointF$ supported on $[0,1]^2$ and for any pair of mechanisms $\mech$ and $\widehat{\mech}$ with allocation rules $\alloc$ and $\hat{\alloc}$ such that
$\alloc(\val)\in [\hat{\alloc}(\val), \hat{\alloc}(\val)+\epsilon]$ for all $\epsilon$,  
we have $\Rev(\jointF; \mech) \geq \Rev(\jointF; \widehat{\mech})-\epsilon$.
\end{restatable}

Having bounded the errors we accumulate due to discretization, we are ready to describe our FPTAS.  The following is a restatement of Theorem~\ref{thm:intro.single.agent.mech} from the introduction.

\begin{restatable}{theorem}{thmfptas}\label{thm:fptas}
For any distribution $\jointF$ supported on $[0,1]^2$, 
for any $\epsilon\in (0,1)$, 
there exists an algorithm with running time $\poly{\frac{1}{\epsilon}}$ that computes a mechanism with revenue at least $\opt(\jointF) - O(\epsilon)$.
\end{restatable}
\begin{proof}
Let $\jointF'$ be the type distribution $\jointF$ with each value rounded down to the nearest multiple of $\epsilon$.
Let $q_i$ be the marginal probability that the value equals $i\cdot \epsilon$,
and let $\jointF'_{i}$ be the distribution over participation costs conditional on the event that the rounded value is $i\cdot \epsilon$.
By \cref{thm:discrete}, we have $\opt(\jointF') \geq \opt(\jointF)-3\epsilon$.
Let $\mech$ be the optimal mechanism for distribution $\jointF'$ with allocation only taking values on the discretization grid. 
By \cref{thm:discrete alloc}, we have 
$\Rev(\jointF'; \mech)\geq \opt(\jointF') -\epsilon$.
Combining the inequalities, we have
\begin{align*}
\Rev(\jointF; \mech)\geq\Rev(\jointF'; \mech)
\geq \opt(\jointF') -\epsilon
\geq \opt(\jointF)-4\epsilon
\end{align*}
where the first inequality holds since the value distribution in $\jointF$ first order stochastically dominates~$\jointF'$.

We now turn to computing the mechanism $\mech$, which we will do by dynamic programming.
For $i,j\leq \frac{1}{\epsilon}$ and $k\leq \frac{1}{\epsilon^2}$, let $R(i,j,k)$ be the optimal revenue from types with value at or below $i\cdot\epsilon$
when the allocation and utility of buyers with value $i\cdot\epsilon$ are $j\cdot \epsilon$ and $k\cdot \epsilon^2$ respectively.
We initialize the matrix by $R(1,j,k) = 0$ for any $k\leq j$ and $R(1,j,k) = -\infty$ for $k>j$. 
For any $i\geq 2$, we have 
\begin{align}
\label{eq:FPTAS.dp}
R(i,j,k) = \max_{j'\leq j} R(i-1,j',k-j) 
+ q_i\cdot \jointF'_{i}(k\cdot \epsilon^2) \cdot 
(i\cdot j-k)\cdot \epsilon^2.
\end{align}
To interpret this expression, we note that first term of \eqref{eq:FPTAS.dp}, $\max_{j' \leq j} R(i-1,j',k-j)$, is the revenue from types with value at most $(i-1)\epsilon$.
Indeed, the allocation of type with value $(i-1)\epsilon$ should not exceed the allocation of type with value $i\cdot\epsilon$, 
and hence the choice of $j'$ is at most $j$.
Moreover, when the allocation and utility of type with value $i\cdot \epsilon$ are $j\cdot \epsilon$
and $k\cdot \epsilon^2$ respectively, 
\cref{lem:revelation} implies that the utility of type with value $(i-1) \epsilon$ is exactly $(k-j)\epsilon^2$ (which is why we have $k-j$ as the third argument), which is independent of the choice of $j'$.

The second term of \eqref{eq:FPTAS.dp} is the expected revenue from types with value at most $i\cdot\epsilon$.
Note that 
$q_i\cdot \jointF'_{i}(k\cdot \epsilon^2)$ is the total probability of those types that have incentives to participate in the auction, 
and $(i\cdot j-k)\cdot \epsilon^2$ is the payment of those types, 
which is derived by subtracting the utility $k\cdot \epsilon^2$ from the expected value $i\cdot j\cdot \epsilon^2$,
if they choose to participate. 

In order to compute the matrix of $R$, we need to compute $\frac{1}{\epsilon} \times \frac{1}{\epsilon} \times \frac{1}{\epsilon^2} = \frac{1}{\epsilon^4}$ numbers,
and computing each number requires taking the maximum over at most $\frac{1}{\epsilon}$ numbers. 
Thus the matrix $R$ can be computed in time $O(\frac{1}{\epsilon^5})$.  

Once $R$ has been filled, note that $\Rev(\jointF'; \mech) = \max_{j,k}R(1/\epsilon,j,k)$.  To recover the allocation rule of $\mech$ from $R$, note that if $j'$ is the value of $j$ maximizing $R(1/\epsilon,j,k)$ then $\mech$ allocates $j\epsilon$ to agents with value $1$.  The rest of the allocation rule can be recovered similarly by unrolling the recursion from \eqref{eq:FPTAS.dp}: for each $i \leq 1/\epsilon$, if $j'$ is the choice that maximizes $R(i-1,j',k-j)$, then agents with value $(i-1)\epsilon$ are allocated $j'\epsilon$.
\end{proof}

\subsection{Menu Complexity}
\label{sec:menu complexity}
By the taxation principle, the optimal mechanism for the single buyer problem 
is essentially posting a menu of allocations and payments
for the buyer to choose from. 
Let $d$ be the size of the support for the cost distribution.
\cref{exp:lottery better} shows that it can be optimal to offer a menu of at least $2$ options to the buyer when $d\geq 2$.
In this section, we provide an upper bound, as a function of $d$, on the number of menu entries required in the optimal mechanism, i.e., on the menu size of the optimal mechanism. 
The main idea is to reduce the problem to $d$ different revenue maximization problems with public budgets. 
The proof is deferred to \cref{sec:proof complexity}.

\begin{restatable}{theorem}{thmmenu}\label{thm:menu}
For selling a single item to a single buyer with private participation cost,
when the size of the support for the cost distribution is $d$,
the menu size required for the optimal mechanism is at most $2d+1$.
\end{restatable}

The following proposition provides a lower bound on the menu complexity even when the participation cost is perfectly correlated with the value. 
This implies that the bound for the menu complexity in \cref{thm:menu} is tight up to a multiplicative factor of 2. 

\begin{proposition}[\citealp{jullien2000participation,li2022selling}]\label{prop:convex correlated}
For selling a single item to a single buyer with private participation cost,
for any $d\geq 1$,
there exists an instance where the participation cost $\cost_{\val}$ is a non-negative, increasing, and strictly convex function of $\val$,
the support sizes of both the marginal value and marginal cost distribution are $d$,
and the menu size required for the optimal mechanism is $d$.\footnote{\citet{li2022selling} considers a single-agent setting about selling information. They characterize the optimal mechanism in their setting by solving a relaxed problem that looks identical to selling a single item with costly participation where the cost of participation is a deterministic and convex function of the valuation. Therefore, we can directly apply the characterization in their relaxed problem to immediately show our lower bound on menu complexity. 
However, note that \citet{li2022selling} does not derive any result when the cost can be stochastic (e.g., independent of the value) or when there are multiple agents.}
\end{proposition}

\subsection{Sample Complexity}
\label{sec:discrete continuous}

In this section, we consider a setting in which the joint distribution $\jointF$ is unknown to the seller, who instead only has access to samples from this distribution, and upper bound the number of samples required to \emph{learn} an up-to-$\epsilon$ optimal auction. 
The proof is deferred to \cref{sec:proof complexity}.

\begin{restatable}{theorem}{thmsample}\label{thm:sample complexity}
For selling a single item to a single buyer with private participation cost, 
for any constant $\epsilon>0$, 
if the joint distribution $\jointF$ is supported on $[0,1]\times[0,1]$,
there exists a mechanism with access to $O(\epsilon^{-6}\log \epsilon^{-1})$ samples
that obtains expected revenue at least $\opt(\jointF) - \epsilon$.
\end{restatable}

Note that in contrast to the traditional single-item auction for a single buyer without participation costs where the sample complexity is $\tilde{O}(\epsilon^{-2})$ \citep[c.f.,][]{guo2019settling}, 
the sample complexity bound in our model is significantly higher. 
The increase of sample complexity originates from two sources.
First, we are learning a two-dimensional correlated distribution instead of a single-dimensional distribution, 
which is intrinsically harder. 
Second, the space of optimal mechanisms we search for is larger. 
Instead of only considering posted pricing mechanisms for auction without participation costs, 
we need to optimize over mechanisms with menu complexity $O(d)$.

%% file: fig/payment_identity.tex
\begin{tikzpicture}[scale = 0.45]

\draw (-0.2,0) -- (12.5, 0);
\draw (0, -0.2) -- (0, 6);

\begin{scope}[thick]
\draw plot [smooth, tension=0.6] coordinates { (0,0) (1.5,2) (5, 3) };
\draw plot [smooth, tension=0.6] coordinates {(5, 3) (8,5) (12,5.5)}; 
\end{scope}

\draw [dotted] (5, 0) -- (5, 3);
\draw (5, -0.7) node {$v_c(x)$};

\draw [dotted] (8, 0) -- (8, 5);
\draw [dotted] (0, 5) -- (8, 5);
\draw (8, -0.7) node {$v$};

\draw (3, 4) node {$p(v)-p_0$};
\draw (3, 1.2) node {$c+p_0$};

\draw (0, 6.7) node {$x(v)$};
\draw (0, -0.7) node {$0$};

\end{tikzpicture}

%% file: fig/payment_identity_c.tex
\begin{tikzpicture}[scale = 0.45]

\draw (-0.2,0) -- (12.5, 0);
\draw (0, -0.2) -- (0, 6);

\begin{scope}[thick]
\draw (0,0) -- (5, 0);
\draw plot [smooth, tension=0.6] coordinates {(5, 3) (8,5) (12,5.5)}; 
\end{scope}

\draw [dotted] (5, 0) -- (5, 3);
\draw (5, -0.7) node {$v_c(x)$};

\draw [dotted] (8, 0) -- (8, 5);
\draw [dotted] (0, 5) -- (8, 5);
\draw (8, -0.7) node {$v$};

\draw (2.7, 3) node {$p_c(v)-p_0$};

\draw (0, 6.7) node {$x_c(v)$};
\draw (0, -0.7) node {$0$};

\end{tikzpicture}

%% file: content/pricing.tex
\section{Optimality of Posted Pricing in the Single-Buyer Setting}
\label{sec:pricing}

Recall that as illustrated in \cref{exp:lottery better}, in general it is not revenue-optimal to post a single take-it-or-leave-it price when buyers have participation costs.  In Theorem~\ref{thm:fptas} we compute an approximately revenue-optimal mechanism, which may involve a menu with multiple lotteries.  But are there conditions under which this complexity is unnecessary, and we recover the simplicity of the posted-price mechanisms that are optimal when participation costs are zero?

In this section, we show that under natural assumptions on the joint distribution, 
posting a deterministic take-or-leave-it price is optimal for the seller. 
We will focus on the case in which the private participation cost is non-negative, 
and by \cref{lem:p0=0}, we set $p_0=0$
and omit it in the notation. 

\subsection{Independent Participation Costs}
In this section, we consider the setting where the cost distribution is independent from the valuation distribution. 
In this case, we provide two sufficient conditions on the valuation distribution such that posted pricing is optimal for revenue maximization. 

\subsubsection{DMR Valuation Distribution}

We start by analyzing the case in which the valuation distribution has decreasing marginal revenue.  We begin with a technical lemma that encapsulates a useful implication of the decreasing marginal revenue property, 
with proof provided in \cref{apx:pricing single}.

\begin{restatable}{lemma}{lemmonotone}\label{lem:monotone}
If the valuation distribution $F$ has decreasing marginal revenue, 
then for any $\cost \geq 0$ and any $\val\geq \cost$, 
we have that $f(\val)(\phi(\val)-\cost)$ is monotone non-decreasing in $\val$. 
\end{restatable}

\begin{restatable}{theorem}{thmdmr}\label{thm:dmr}
With independent values and participation costs, and with non-negative participation costs, 
if the value distribution~$F$ has
decreasing marginal revenue and bounded support,\footnote{The boundedness assumption is to avoid a situation in which the optimal mechanism is not well defined,
i.e., for any mechanism, there exists another mechanism with strictly higher revenue.
For example, if the buyer has participation cost $0$, and the CDF of the value is $F(v) = 1-\frac{1}{v+1}$ for $v\geq 0$, 
the supreme of the revenue is only attained in the limit by offering a take-or-leave-it price $p\to\infty$.
Note that bounded support is a sufficient condition to guarantee this property, but it may not be necessary. 
} 
there exists a revenue-maximizing mechanism that is a posted price mechanism.
\end{restatable}

The details of the proof is provided in \cref{apx:pricing single}.
Intuitively, DMR is an assumption stating that the seller can always extract higher revenue (or equivalently increase virtual welfare) by allocating the item to the agent with higher value instead of lower value when there is no participation cost. 
In \cref{thm:dmr}, we show that the extra negative impact of the participation cost on revenue is decreasing as we reduce the probability of selling the item to a given agent (this is transparent from \cref{thm:revenue equivalence}). 
So by shifting the allocation from lower-value agents to higher-value agents, we simultaneously increase the virtual welfare and decrease the negative impact of participation costs (as the agent will participate less often). 
Shifting the allocation in this way as much as possible ultimately leads to a step function corresponding to posted pricing, which must then be optimal.

\begin{remark} In this section we have focused on the case in which the participation cost is non-negative. 
In \cref{apx:negative outside}, we show that a deterministic  mechanism can guarantee at least 50\% of the optimal revenue when the participation cost could be positive or negative.  Note that, in this context, a deterministic mechanism 
is equivalent to setting a price $p_0$ for participating the auction, 
and an additional price $p$ for winning the item.  The presence of the additional participation cost $p_0$ is intuitive: a negative participation cost means that an agent has a strict incentive to participate in the auction regardless of outcome, and $p_0$ serves to extract a portion of that participation utility as revenue.
\end{remark}

\subsubsection{MHR Valuation Distribution}

We move on to analyzing the case in which the valuation distribution has monotone hazard rate.
Note that if the valuation distribution $F$ satisfies monotone hazard rate, 
the derivative of the virtual value is $\phi'(\val)\geq 1$ 
for all $\val\geq 0$,
and hence distribution~$F$ is regular as well.

\begin{lemma}[\citealp{devanur2017optimal}]\label{lem:regular imply dmr for small v}
If the valuation distribution $F$ is regular, then $f(\val)\virtual(\val)$ is non-decreasing in~$\val$ for $\val \in [0, \val^*]$ where $\val^* = \inf_{\val} \{\virtual(\val) \geq 0\}$. 
\end{lemma}

\begin{restatable}{theorem}{thmmhr}\label{thm:mhr}
With independent values and participation costs, and with non-negative participation costs, 
if the value distribution $F$ has monotone hazard rate, 
there exists a revenue optimal mechanism that is a posted price mechanism. 
\end{restatable}
The details of the proof is provided in \cref{apx:pricing single}.
%
Note that there is asymmetry between the conditions we require on the value distribution and the cost distribution to establish optimality of posting a price. 
We have shown that, for any distribution over participation costs, posting a price is optimal as long as the value distribution satisfies monotone hazard rate or decreasing marginal revenue.
However, even if the cost distribution is uniform, 
posting a fixed take-it-or-leave-it price need not be optimal without further assumptions on the value distribution. The following example illustrates such a scenario.

\begin{example}
The value of the buyer is $1$ or $2$ with probability $\frac{1}{2}$ each. 
The participation cost is uniformly distributed in $[0,1]$.
The optimal posted price mechanism is to post a price of $1$, generating revenue $\frac{1}{2}$. 
However, the mechanism that offers the menu of two probability-price lotteries $(1, 1), (\frac{2}{3}, \frac{1}{3})$ 
has revenue $\frac{5}{9}$. 
The multiplicative gap between the optimal mechanism and optimal pricing is therefore higher than $1.111$.
\end{example}

\subsection{Perfectly Correlated Participation Costs}
\label{sec:correlated}

In this section, we assume the participation cost is perfectly correlated with the value. 
That is, there is a publicly known mapping from each value $\val$ to the unique corresponding participation cost $\cost_{\val}$, and moreover $\cost_{\val}$ 
is monotone increasing in $\val$.\footnote{The case in which the participation cost $\cost_{\val}$ is monotone non-increasing in $\val$
is rather trivial since the individual rationality constraint only binds for the lowest type.} 
Note that as illustrated in \cref{prop:convex correlated},
when the participation cost is convex in the value,
the optimal mechanism may not be posted pricing. 
In fact, the menu complexity can be infinite for continuous valuation distributions. 
In contrast, when the participation cost $\cost_{\val}$ is concave in~$\val$, 
posted pricing is always revenue maximizing.
The proof of the following theorem is provided in \cref{apx:pricing single}.

\begin{restatable}{theorem}{thmCorrelate}\label{thm:correlate menu}
If the participation cost $\cost_{\val}$ is a non-negative and concave function of~$\val$,
then for any value distribution $F$,
there exists a revenue optimal mechanism that is a posted price mechanism.
\end{restatable}

\section{Mechanisms for Multiple Buyers}
\label{sec:multi}
\subsection{An Approximately Optimal Mechanism for Multiple Buyers}
\label{sec:multi-const}

For the multi-buyer setting, as illustrated in \cref{sec:model} (and in \cref{sec:revelation}), optimizing revenue involves not only the mechanism's allocation rule but also the equilibrium participation behavior of the agents.  As with the single-buyer case, revenue maximization is inherently non-convex in this environment (see \cref{sub:nonconvex} for further discussion).

In this section, we provide a polynomial time algorithm for computing a mechanism that is a constant approximation to the optimal revenue.  Our solution will take the form of a sequential opt-out-or-revelation mechanism in which the agents' participation decisions are unambiguous.\footnote{Each agent will be offered a menu of allocations and payments.  While the menu offered to agent $i$ can depend on the behavior of others, the agent's outcome will be otherwise independent of other agents' decisions given the menu.  Agent $i$ will therefore opt in precisely when the expected utility from the menu exceeds the participation cost.}  We will describe our algorithm for continuous type distributions.  At the end of this section we discuss how to adapt our procedure to discrete distributions provided as an explicit assignment of probabilities to a finite collection of types.  The following restates Theorem~\ref{thm:intro.many.agent.mech} from the introduction.

\begin{restatable}{theorem}{thmpolymulti}\label{thm:poly_multi}
For any product distribution $\jointF=\times_{i=n}\jointF_i$ with $\jointF_i$ supported on $[0,1]^2$
with density of the conditional participation cost being at most $\eta\geq 1$, 
for any $\epsilon\in (0,1)$, 
there exists an algorithm with running time $\poly{\frac{n\eta}{\epsilon}}$ 
that computes a mechanism with revenue at least $\frac{1}{2}\opt - \epsilon$.
\end{restatable}

As discussed in the introduction, the main idea of our construction is to reduce to the single-buyer revenue maximization problem, employing a reduction framework introduced by~\citet{alaei2014bayesian}.  To do this we consider the ex-ante relaxation of the optimal mechanism, 
efficiently compute the ex-ante optimal mechanism for each buyer, 
and sequentially offer those mechanisms to the buyers. 
One challenge for adopting this approach is that 
given an ex-ante constraint in the single-buyer problem, 
the efficient computation result through dynamic program only finds the best fixed mechanism, 
while it might be necessary for the seller to randomize over mechanisms. 
That is, the seller chooses a distribution over mechanisms, 
informs the buyer about the realization, 
offers the buyer the realized mechanism, 
and then the buyer makes the participation decision. 
Such randomization cannot be collapsed as a single mechanism with randomized allocations since it affects the buyer's participation decisions 
and eventually the expected revenue. 
The issue of requiring randomized mechanisms also arises 
when discretizing the allocation spaces. 
If we only consider a fixed mechanism, even if randomized allocation rules are allowed in the fixed mechanism, 
increasing the allocation probabilities to the discretized grid may violated the interim feasibility constraints, 
while decreasing the allocation probabilities may exclude the buyer from participating in the auction. 
Fortunately, such issue of randomization can be alleviated by showing that in order to approximate the optimal ex-ante revenue, it is sufficient to consider the randomization over two fixed mechanisms over a discretized space of allocation rules. 
We formalize this in Lemma~\ref{lem:ex-ante fixed} below, but first we describe the way to discretize the type distributions.

Formally, given a type distribution $\jointF$, we will approximate $\jointF$ via the following sequence of rounding operations.  First consider the discretization grid $\grid_0=\{0,\epsilon,2\epsilon,\dots,1\}$. 
Given type distribution $\jointF$, let $\jointF'$ be the distribution that rounds all values down to the discretization grid $\grid_0$ (i.e., down to the nearest multiple of $\epsilon$).
Write $z_k$ for the marginal probability that the rounded value equals $k \epsilon$,
and let $\jointF'_{k}$ be the distribution over participation costs conditional on the event that the rounded value is $k\epsilon$.
Moreover, consider an additional distribution~$\jointF\dprimed$ that is like $\jointF'$ but with two changes.  First, each marginal value probability $z_k$ is rounded \emph{up} to the nearest multiple of $\epsilon^2$; call these rounded probabilities $z\dprimed_k$.  Second, for each conditional cost distribution $\jointF'_{k}$, we will round all participation costs up to 
the nearest multiple of $\epsilon^2$.
Write $\jointF\dprimed_k$ for the resulting rounded conditional cost distribution.\footnote{Note that distribution $\jointF\dprimed$ is not a well defined distribution since the total probability measure may exceed 1. 
However, the expected revenue $\Rev(\jointF\dprimed; \mech)$ given any mechanism $\mech$ is well defined.}  Write $\grid_1$ for the discretization grid consisting of $\{0,\epsilon^2,2\epsilon^2,\dots,1\}$,
so that each $\jointF\dprimed_k$ is supported on $\grid_1$.
Let $\opt_{\quant}(\jointF)$ be the optimal revenue for distribution $\jointF$ given ex-ante allocation constraint~$\quant$
and let $\optd_{\quant}(\jointF)$ be the optimal revenue for distribution $\jointF$ given ex-ante constraint $\quant$,\footnote{An ex-ante allocation constraint is an upper bound on the probability that the agent will receive the item, in expectation over all sources of randomness including the agent's own type.}
only randomizing over two fixed mechanisms with allocations
and ex-ante probabilities on the discretized grid $\grid_0$. 

Now that we have described our intended rounding of the type distribution, 
we can show that the errors obtained under this rounding is small when we restrict our mechanism to have ex-ante sale probabilities that are multiples of $\epsilon$.
Moreover, such rounding errors allow us to compute the profile of optimal ex ante probabilities efficiently. 
The proof of the following lemma is provided in \cref{sub:multi-proof}.


\begin{restatable}{lemma}{lemAproxComp}\label{lem:compute_approx}
For any $\hat{\epsilon}>0$,
there exists an algorithm with running time $\poly{\frac{n}{\hat{\epsilon}}}$ that computes a profile of $\{\quant\dprimed_i\}_{i\in[n]}$ subject to the constraint that $\sum_{i\in[n]} \quant\dprimed_i\leq 1$,
and the corresponding mechanisms $\{\mech_{i,\quant\dprimed_i}\}_{i\in[n]}$
such that $\{\quant\dprimed_i\}_{i\in[n]}$
maximizes $\sum_{i\in[n]}\Rev(\jointF_i; \mech_{i,\quant\dprimed_i})$
and 
\begin{align*}
\Rev(\jointF_i; \mech_{i,\quant\dprimed_i}) \geq \optd_{\quant\dprimed_i}(\jointF\dprimed_i) - \hat{\epsilon}, \quad \forall i.
\end{align*}
\end{restatable}

\begin{proof}[Proof of Theorem~\ref{thm:poly_multi}]
By \cref{lem:compute_approx}, 
for any $\hat{\epsilon}>0$,
with running time $\poly{\frac{n}{\hat{\epsilon}}}$, 
we compute the profile of $\{\quant\dprimed_i\}_{i\in[n]}$
and the corresponding mechanisms $\{\mech_{i,\quant\dprimed_i}\}_{i\in[n]}$
such that 
\begin{align*}
\sum_{i\in[n]}\Rev(\jointF_i; \mech_{i,\quant\dprimed_i}) \geq 
\sum_{i\in[n]}\optd_{\quant\dprimed_i}(\jointF\dprimed_i) - n\hat{\epsilon}.
\end{align*}

Finally, to complete the proof of Theorem~\ref{thm:poly_multi} we must show how to combine the single-agent mechanisms $\mech_{i,\quant\dprimed_i}$ to construct a  multi-agent mechanism with similar total revenue.
Fix any order of the buyers, and consider a sequential mechanism that 
attempts to sell the item to each buyer $i$ in order.
As long as the item has not yet been sold, each buyer $i$ will be offered the mechanism (i.e., the allocation rule) $\mech_{i,\quant\dprimed_i}$ with 
probability~$\frac{1}{2 - \sum_{j < i}\quant\dprimed_j}$ (and with the remaining probability buyer $i$ will not be given an opportunity to participate).
We claim that under this procedure, each buyer~$i$ is offered the mechanism $\mech_{i,\quant\dprimed_i}$ with probability between 
$1/2 - n \eta  \hat{\epsilon}$ 
and $1/2$, and hence the total probability that buyer $i$ obtains the item is at least 
$(\frac{1}{2}  - n \eta  \hat{\epsilon}) \quant\dprimed_j$ and at most $\frac{1}{2}(\quant\dprimed_i + \eta\hat{\epsilon})$.  
Indeed, by induction, each buyer $j < i$ receives the item with total probability at most 
$\frac{1}{2}(\quant\dprimed_j + \eta\hat{\epsilon})$, 
and hence
the probability that the item is unsold when buyer $i$ is to be approached is at least 
$1 - \frac{1}{2}\sum_{j<i}(\quant\dprimed_j + \eta \hat{\epsilon})$. 
Taking into account the probability that we offer anything to buyer $i$ when the item is unsold, we conclude that buyer $i$ is offered mechanism $\mech_{i,\quant\dprimed_i}$ with probability at least 
$\frac{1 - \frac{1}{2}\sum_{j<i}(\quant\dprimed_j + \eta \hat{\epsilon})}{2-\sum_{j<i}\quant\dprimed_j} \geq 1/2 - n  \eta  \hat{\epsilon}$ 
and at most 
$\frac{1 - \frac{1}{2}\sum_{j<i}\quant\dprimed_j}{2-\sum_{j<i}\quant\dprimed_j} = 1/2$, 
as claimed.
Therefore, the expected revenue from this sequential mechanism is at least 
\begin{align*}
\left(\frac{1}{2} - n\eta\hat{\epsilon}\right)\sum_{i\in[n]}\Rev(\jointF_i; \mech_{i,\quant\dprimed_i})
&\geq \left(\frac{1}{2} - n\eta\hat{\epsilon}\right)\sum_{i\in[n]}\rbr{\optd_{\quant\dprimed_i}(\jointF\dprimed_i) - n\hat{\epsilon}}\\
&\geq \left(\frac{1}{2} - n\eta\hat{\epsilon}\right)\rbr{\max_{\{\quant_i\}_{i\in[n]}}\sum_{i\in[n]}\opt_{\quant_i}(\jointF_i) - 6n\hat{\epsilon}}\\
&\geq \left(\frac{1}{2} - n\eta\hat{\epsilon}\right)\opt - 3n\hat{\epsilon} \geq \frac{1}{2}\opt - 4n\eta\hat{\epsilon}
\end{align*}
where the second inequality holds by applying \cref{lem:ex-ante fixed}.
For any $\epsilon > 0$, 
by setting 
$\hat{\epsilon} = \frac{\epsilon}{4n\eta}$, 
the expected revenue loss is $\epsilon$, 
and the running time is 
$\poly{\frac{n\eta}{\epsilon}}$.
\end{proof}

\paragraph{Remark: Relationship to Correlation Gap} Note that the technique of correlation gap \citep{yan-11,alaei2013simple} for obtaining the approximation of $\sfrac{e}{(e-1)}$ 
does not apply here directly
since the buyers does not satisfy the expected utility representation given the allocation and payment rules, 
and the ex-ante optimal mechanisms for each agent offers complex menus instead of posting prices. 

\paragraph{Remark: Continuous versus Discrete Type Distributions} Our analysis above introduces runtime a dependency on $\eta$, the maximum density of the conditional participation cost in type distribution $\jointF$.  This dependency arises because of errors in estimated agent utility that are introduced when rounding values to the nearest multiples of $\epsilon$, which was needed to define our dynamic program over allocation rules.  Of course, $\eta$ is well-defined only for continuous type distributions.  Alternatively, if each agent $i$'s type distribution were listed explicitly as a discrete set of types $T_i$ and their associated probabilities, then instead of rounding values to multiples of $\epsilon$ we could replace the value grid $\Psi_0$ with the (at most) $|T_i|$ values in the support of agent $i$'s distribution.  Doing so would remove the errors in our utility estimates (since we are using the exact values) and therefore avoid any dependency on $\eta$. However, this introduces a polynomial runtime dependence on $\max_i |T_i|$ since our dynamic program for agent $i$ will need to include $|T_i|$ values.  Also, since the distribution is not continuous it might happen that there is a non-vanishing fraction of agent types who are indifferent between participating and not participating in the resulting allocation rule, but in this case the rule can be perturbed by an arbitrarily small amount to make the participation probability unambiguous.  The end result is that in Theorem~\ref{thm:poly_multi} we could replace $\eta$ with $\max_i |T_i|$ if each distribution $\jointF_i$ is a discrete distribution over $T_i$ types.

\subsection{Approximate Optimality of Posted Pricing}
\label{sec:multi-pricing}

In this section, we show that with additional assumptions on buyers' type distributions, 
sequential posted pricing is approximately optimal for revenue maximization. 
The following lemma provides a reduction framework for single-item environments 
that lifts approximation results for posted pricing 
from a single buyer to multiple buyers, 
and so avoids the direct analysis of interaction among different buyers.

\begin{lemma}[\citealp{feng2020simple}]\label{thm:multi-buyer reduction}
For a single-item environment, 
if there exists $\gamma\geq 1$ such that for each buyer $i$, for any constraint on the sale probability $q_i\leq 1$, 
the approximation ratio of (possibly randomized) posted pricing for the single-buyer problem
given sale probability constraint $q_i$ is at most~$\gamma$,
then the approximation ratio of sequential posted pricing is at most $\sfrac{\gamma e}{(e-1)}$ for the multi-buyer problem. 
\end{lemma}

Note that when agents have private participation costs, 
it is important to specify the timeline for the agents to make participation decisions, 
as different timeline will affect the equilibrium choice of the agents for participation. 
In this section, for sequential posted pricing mechanisms, 
we assume that each agent only needs to make the participation decision after seeing the realized price offered by the seller.

Since \cref{thm:multi-buyer reduction} allows us to reduce the analysis to the single-buyer problem with allocation constraints, 
in the following two Lemmas, 
we will prove such approximation results for both independent participation costs 
and perfectly correlated participation costs.
Both of the proofs are provided in \cref{sub:multi-pricing}.

\begin{restatable}{lemma}{lemDmrSupply}\label{lem:dmr supply}
In the single-buyer setting with
independent and non-negative participation costs,
if the value distribution $F$ has
decreasing marginal revenue and bounded support,
for any mechanism that sells the item with probability $q\in [0, 1]$, 
there exists a (possibly randomized) posted price mechanism with weakly higher expected revenue that sells the item with probability at most $q$. 
\end{restatable}

\begin{restatable}{lemma}{lemCorrelateSupply}\label{thm:correlate}
In the single-buyer setting with
the participation cost being a non-negative and concave function of $\val$,
for any mechanism that sells the item with probability $q\in [0, 1]$, 
there exists a (possibly randomized) posted price mechanism with at least half of the expected revenue that sells the item with probability at most $q$. 
\end{restatable}
Combining \cref{thm:multi-buyer reduction} with \cref{lem:dmr supply,thm:correlate}, 
we have the following approximation result.

\begin{theorem}
In the multi-buyer setting, if for each buyer both the participation cost is non-negative and either of the following holds:
\begin{itemize}
\item
the participation cost is independent of the value,
and the value distribution $F$ has decreasing marginal revenue and bounded support; or
\item
the participation cost is a concave function of the value;
\end{itemize}
then there is a sequential posted pricing mechanism that generates at least a $\tfrac{1}{2}(1-\sfrac{1}{e})$
fraction of the optimal revenue. If the former condition holds for each buyer, then the guaranteed fraction of the optimal revenue is $1-\sfrac{1}{e}$.
\end{theorem}

%% file: content/conclusion.tex
\section{Discussions}
\label{sec:conclude}
In this paper we characterize the optimal mechanisms when buyers have private participation costs, show how to compute approximately revenue-optimal mechanisms in polynomial time,
and provide sufficient conditions on the value distributions or the cost distributions
such that posted pricing is optimal or approximately optimal. 
In this section, we will provide further discussions on the alternative models for costly participation and the remaining open questions. 

\subsection{Alternative Communication Models}
In this paper we focus on the class of sequential opt-out-or-revelation mechanisms, in which 
the seller has the capability to send messages to the buyers
before they choose whether to participate. 
This advance communication is useful for coordinating the buyers' participation decisions.
One might naturally wonder about alternative models that differ in how much communication is allowed before participation decisions are made.
We enumerate some possibilities here, each suitable for different applications.

\paragraph{Two-Way Communication} The buyers and the seller can communicate repeatedly in an unrestricted manner before participation decisions are made.  For example, the seller can elicit private value and cost information from the buyers and provide recommendations for participation.  A generalized revelation principle by \citet{myerson1982optimal} is applicable in this context.

\paragraph{One-Way Communication} The seller can send messages to the buyers before participation decisions are made, but a buyer cannot communicate with the seller before paying the cost of participation.  The sequential opt-out-or-revelation mechanisms studied in this paper fall within this model.

\paragraph{No Communication} Neither the seller nor the buyers can communicate prior to participation.  All buyers must make their participation decisions simultaneously.  In this case, the classic revelation principle by \citet{myerson1981optimal} applies.

\bigskip


In the single-buyer setting, these three models coincide.  As a result, all the results presented in \cref{sec:single,sec:pricing} apply seamlessly in any of these models. In the multi-buyer setting, however, a more restrictive communication model 
reduces the space of feasible mechanisms, 
which has the potential to reduce optimal revenue.
Nevertheless, it turns out that the mechanisms we devised in \cref{sec:multi} -- which employ one-way communication -- are also approximately optimal with respect to the broader class of mechanisms in the two-way communication model.
This is because the approximation mechanisms we devised in \cref{sec:multi} remain approximately optimal when compared to the ex-ante relaxation benchmark. This benchmark serves as an upper bound for the optimal revenue even in the two-way communication model. Therefore, the approximation mechanisms proposed in \cref{sec:multi} are still approximately optimal when two-way communication is allowed.

On the other hand, our designed mechanisms heavily rely on the ability to make sequential participation decisions, rendering them infeasible in the no-communication model. An important future direction to explore is whether it's possible to compute an (approximately) optimal mechanism in the no-communication model for multiple buyers in polynomial time. Additionally, it's intriguing to investigate whether the revenue gap between the no-communication model and the two-way communication model is at most a constant.


\subsection{Open Questions}
Following our work, 
many interesting questions remain open even in the single-buyer setting, including:
\begin{enumerate}
\item For the setting with independent item values and non-negative participation costs,
is posted pricing a constant approximation to the optimal revenue without additional assumptions?

\item Can the exact optimal mechanism be computed in time polynomial in the size of the support?

\item For multi-item settings (e.g., unit-demand valuations) with non-negative participation costs, 
do there exist simple mechanisms (e.g., item pricing) that generate a constant fraction of the optimal revenue,
under assumption such as DMR or MHR on the item value distributions?
\end{enumerate}


%% file: content/appendix.tex
\section{A Generalized Revelation Principle}
\label{sec:revelation}

Let $\hat{\reals} \equiv \reals^2\cup \{\emptyset\}$ be the augmented outcome space, 
where $\emptyset$ represents not participating the auction. 
Hence the utility of any buyer $i$ for outcome $\emptyset$ is $0$.

For the private participation cost setting, 
consider any interactive protocol between the seller and the buyer. 
For any $i$, let $\order_i$ be the $i$th buyer that first interacts with the seller. 
Let $\info_i$ be the message the seller send to buyer $\order_i$
before the interaction. 
Similar to the argument in \cite{myerson1981optimal},
any mechanism can be transformed into a generalized version of revelation mechanism 
$\widehat{\mech}:\reals^{2n} \to \hat{\reals}^n$
where all buyers participate the auction.
Moreover, conditional on information~$\info_i$,
for any buyer~$\order_i$ and any~$(\val_{\order_i}, c_{\order_i})$, 
if there exists $(\val_{-\order_i},c_{-\order_i})$ such that $\widehat{\mech}((\val_{\order_i},\val_{-\order_i}),(c_{\order_i},c_{-\order_i})) = \emptyset$, 
then $\widehat{\mech}((\val_{\order_i},\val'_{-\order_i}),(c_{\order_i},c'_{-\order_i})) = \emptyset$ for all $(\val'_{-\order_i},c'_{-\order_i})$. 
Note that since the buyers' utility conditional on participation is invariant of the participation cost, 
for any buyer~$\order_i$ with value~$\val_{\order_i}$, 
for any pair of participation cost $(c_{\order_i}, c'_{\order_i})$
such that $\widehat{\mech}((\val_{\order_i},\val_{-\order_i}),(c_{\order_i},c_{-\order_i})) \neq \emptyset$
and $\widehat{\mech}((\val_{\order_i},\val_{-\order_i}),(c'_{\order_i},c_{-\order_i})) \neq \emptyset$ for some $(v_{-\order_i}, c_{-\order_i})$, 
we have 
\begin{align*}
& \expect[v_{-\order_i}, c_{-\order_i}]{u_{\order_i}(\widehat{\mech}((\val_{\order_i},\val_{-\order_i}),(c_{\order_i},c_{-\order_i}))) \given \info_i} \\
=\,& \expect[v_{-\order_i}, c_{-\order_i}]{u_{\order_i}(\widehat{\mech}((\val_{\order_i},\val_{-\order_i}),(c'_{\order_i},c_{-\order_i}))) \given \info_i}.
\end{align*}

Finally, it is sufficient to show that this mechanism can be implemented as a sequential opt-out-or-revelation mechanism in our setting.
For any order $i$ and corresponding buyer~$\order_i$, let $\jointF'_{\order_i}$ be the distribution over $(\val_{\order_i}, c_{\order_i})$
conditional on 
$\widehat{\mech}((\val_{\order_i},\val_{-\order_i}),(c_{\order_i},c_{-\order_i})) = \emptyset$ for some (equivalently all) $(\val_{-\order_i},c_{-\order_i})$.
Moreover, let $G_{\order_i}(v_{\order_i})$ be the distribution over participation costs $c_{\order_i}$ conditional on value $\val_{\order_i}$ and $\widehat{\mech}((\val_{\order_i},\val_{-\order_i}),(c_{\order_i},c_{-\order_i})) \neq \emptyset$ for some (equivalently all) $(\val_{-\order_i},c_{-\order_i})$.
Recall that $\bar{\reals} \equiv \reals\cup\{\psi\}$ is the augmented report space
where $\psi$ represents not participating the auction.
Let $\mech$ be the sequential opt-out-or-revelation mechanism 
that interacts with buyers according to order $\order$,
sends information $\info_i$ to buyer $\order_i$ before the interaction,
and maps the reported valuations to the profile of allocations and payments constructed as follows. 
For any buyer $\order_i$ not participating the auction, or equivalently reporting $\psi$, 
the allocation and payment for buyer $i$ is zero. 
Moreover, when buyer $i$ reports $\psi$, 
mechanism~$\mech$ resamples $(\val_{\order_i},c_{\order_i})$ according to distribution 
$\jointF'_{\order_i}$, 
and when buyer $i$ reports $\val_{\order_i}$, 
mechanism~$\mech$ resamples $c_{\order_i}$ according to distribution 
$G_{\order_i}(\val_{\order_i})$. 
Then mechanism $\mech$ runs~$\widehat{\mech}$ on the generated type profile of all $n$ buyers. 
It is easy to verify that for the mechanism~$\mech$, 
there exists an equilibrium where 
each buyer participates (i.e., does not play $\psi$) in $\mech$ and reports truthfully on her valuation if and only if her outcome is not $\emptyset$ in $\widehat{\mech}$. 
Moreover, under mechanism~$\mech$, this equilibrium generates the same distribution over outcomes for all buyers 
as the truthful equilibrium of~$\widehat{\mech}$, and therefore the same revenue as in~$\widehat{\mech}$.

\section{Missing Proof for Characterization and Computation}
\label{sec:proof complexity}
\lemRevEquiv*
\begin{proof}
For any participation cost $\cost$, 
the buyer with cost $\cost$ participates in the mechanism if and only if $\val\geq \val_{\alloc}(\cost)$. 
First we consider the case $c+p_0> 0$, 
in which case $\val_{\alloc}(\cost)> 0$.
Moreover, the payment of a buyer having any value $\val\geq \val_{\alloc}(\cost)$ is 
$\val \cdot \alloc(\val) - \int_0^{\val} \alloc(z) \dd z + p_0
= \val \cdot \alloc_{\cost}(\val) - \int_0^{\val} \alloc_{\cost}(z) \dd z - c$, 
and the payment of a buyer having any value $\val< \val_{\alloc}(\cost)$ is $0$
since such a buyer does not participate the auction. 
This is illustrated in \Cref{fig:payment identity}.
Note that by \citet{myerson1981optimal}, 
we have 
\begin{align*}
\expect[\val\sim \jointF_c]{\val\cdot \alloc_{\cost}(\val) - \int_0^{\val} \alloc_{\cost}(z) \dd z}
= \expect[\val\sim \jointF_c]{\alloc_{\cost}(\val)\virtual_c(\val)}
\end{align*}
and thus the expected revenue given cost $\cost\geq -p_0$ is 
\begin{align*}
&\expect[\val\sim \jointF_c]{\val\cdot \alloc_{\cost}(\val) - \int_0^{\val} \alloc_{\cost}(z) \dd z
- \cost \cdot \indicate{\val \geq \val_{\alloc}(\cost)}}\\
=\,& \expect[\val\sim \jointF_c]{\alloc_{\cost}(\val)\virtual_c(\val)}
- (1-\jointF_c(\val_{\alloc}(\cost)))\cdot \cost,
\end{align*}
where ${\bf 1}$ is the indicator function.
For the case $c \leq -p_0$, 
we have $\val_{\alloc}(\cost) = 0$.
Thus the payment of a buyer having any value $\val\geq 0$ is
$\val \cdot \alloc(\val) - \int_0^{\val} \alloc(z) \dd z + p_0
= \val \cdot \alloc_{\cost}(\val) - \int_0^{\val} \alloc_{\cost}(z) \dd z + p_0$, 
and again by \citet{myerson1981optimal}, 
the expected revenue given cost $\cost$ is 
\begin{align*}
\expect[\val\sim \jointF_c]{\val\cdot \alloc_{\cost}(\val) - \int_0^{\val} \alloc_{\cost}(z) \dd z + p_0}
= \expect[\val\sim \jointF_c]{\alloc_{\cost}(\val)\virtual_c(\val)}
+p_0.
\end{align*}
Taking expectation over $\cost\sim G$
gives us the desired characterization.
\end{proof}

\lempzz*
\begin{proof}
Suppose first that $\pay_0 < 0$.  Given any allocation $\alloc$ and 
letting $\val' = \sup_{\val} \{\val \cdot \alloc(\val) - \int_0^{\val} \alloc(z) \dd z \leq \pay_0\}$,
define 
\begin{align*}
\hat{\alloc}(\val) &= \begin{cases}
\alloc(\val) & \val \geq \val'\\
\alloc(\val') & \val < \val',
\end{cases}\\
\hat{\pay}(\val) &= \val \hat{\alloc}(\val) - \int_0^{\val} \hat{\alloc}(z)\dd z.
\end{align*}
One can verify that the revenue of the seller is weakly higher with allocation $\hat{\alloc}$ and payment~$\hat{\pay}$
since the types in mechanism with $\alloc$ and $\pay$ with negative payment will not participate under the new mechanism 
while the types with positive payment will participate and pay the same amount. 

When $\pay_0 > 0$ and the participation costs are non-negative, 
given any allocation and payment $\alloc, \pay$, 
we can instead define 
\begin{align*}
\hat{\alloc}(\val) &= \begin{cases}
\alloc(\val) & \val \geq \sup_{\val} \{\int_0^{\val} \alloc(z) \dd z \leq \pay_0\}\\
0 & \val < \sup_{\val} \{\int_0^{\val} \alloc(z) \dd z \leq \pay_0\},
\end{cases}\\
\hat{\pay}(\val) &= \val \hat{\alloc}(\val) - \int_0^{\val} \hat{\alloc}(z)\dd z.
\end{align*}
The revenue of the seller is then unchanged if we use the mechanism with allocation $\hat{\alloc}$ and payment $\hat{\pay}$, instead of $\alloc$ and $\pay$, 
since all types of the buyer have the same allocation and payment in best response. 
\end{proof}

\thmdiscreteerror*
\begin{proof}
Due to the symmetry between $t_1$ and $t_2$, it is sufficient to show that $\opt(t_1) - \opt(t_2) \leq 3\epsilon$.
For simplicity, we write $t_k(\omega)$ for the pair $(\val_k(\omega), \cost_k(\omega))$.
Note that by \Cref{lem:revelation},
it is sufficient to assume $\alloc(\val), \pay(\val)$ are the optimal allocation and payment function
with parameter $p_0$. 
Let $\hat{\alloc}(\val) = \alloc(\val+\epsilon)$
and $\hat{\pay}(\val) = \val\hat{\alloc}(\val) - \int_0^\val \hat{\alloc}(z) \dd z - 2\epsilon + p_0$. 
Next we show that $\Rev(t_2 ; \hat{\alloc}, \hat{\pay}) \geq \opt(t_1) - 3\epsilon$.
Since $\hat{\alloc}$ is non-decreasing, 
$\hat{\alloc}$ and $\hat{\pay}$ are incentive compatible for the buyer without participation cost, 
and the buyer will truthfully reveal her valuation to the mechanism if the utility for participation is at least her participation cost. 
Note that 
\begin{align*}
\hat{\pay}(\val) &= \val\hat{\alloc}(\val) - \int_0^\val \hat{\alloc}(z)\ \dd z - 2\epsilon + p_0 
= \val\alloc(\val + \epsilon) - \int_0^\val \alloc(z+\epsilon)\ \dd z - 2\epsilon + p_0\\
&= \val\alloc(\val + \epsilon) - \int_\epsilon^{\val+\epsilon} \alloc(z)\ \dd z - 2\epsilon + p_0
\geq \val\alloc(\val + \epsilon) - \int_0^{\val+\epsilon} \alloc(z)\ \dd z - 3\epsilon + p_0 \\
&= \pay(\val+\epsilon) - 3\epsilon.
\end{align*}
The above inequality holds since $\alloc(z) \leq 1$ for any $z$.
Moreover, we have 
\begin{align*}
\util(\val; \hat{\alloc}, \hat{\pay})
&= \int_0^\val \hat{\alloc}(z)\ \dd z + 2\epsilon + p_0
= \int_0^\val \alloc(z+\epsilon)\ \dd z + 2\epsilon + p_0\\
& \geq \int_0^{\val+\epsilon} \alloc(z)\ \dd z + \epsilon + p_0
= \util(\val+\epsilon; \alloc, \pay) + \epsilon
\end{align*}
where the inequality holds again because  $\alloc(z) \leq 1$ for any $z$.
Finally, combining the inequalities, 
we have 
\begin{align*}
\Rev(t_2 ; \hat{\alloc}, \hat{\pay}) 
&= \int_{\Omega} \hat{\pay}(\val_2(\omega))
\cdot \indicate{\util(\val_2(\omega); \hat{\alloc}, \hat{\pay}) \geq \cost_2(\omega)}\ \dd P(\omega)\\
&\geq \int_{\Omega} (\pay(\val_2(\omega)+\epsilon)-3\epsilon)
\cdot \indicate{\util(\val_2(\omega)+\epsilon; \alloc, \pay) + \epsilon \geq \cost_2(\omega)}\ \dd P(\omega)\\
&\geq \int_{\Omega} (\pay(\val_1(\omega))-3\epsilon)
\cdot \indicate{\util(\val_1(\omega); \alloc, \pay) \geq \cost_1(\omega)}\ \dd P(\omega)
\geq \opt(t_1) - 3\epsilon.
\end{align*}
The first inequality holds by combining the above inequalities, 
and the second inequality holds since $\inftynorm{t_1(\omega) - t_2(\omega)} \leq\epsilon$
and both $\pay$ and $\util$ are monotone in $\val$. 
\end{proof}

\thmdiscretealloc*
\begin{proof}
Let $\pay,\hat{\pay}$ and $\util,\hat{\util}$ be the payment functions and utility functions
in mechanisms $\mech$ and $\widehat{\mech}$ respectively. 
By \cref{lem:revelation}, we have 
\begin{align*}
\pay(\val) &= \val \alloc(\val) - \int_0^{\val} \alloc(z) \dd z 
\geq \val \hat{\alloc}(\val) - \int_0^{\val} (\alloc(z)+\epsilon) \dd z 
\geq \hat{\pay}(\val) - \epsilon
\intertext{and}
\util(\val) &=\int_0^{\val} \alloc(z) \dd z 
\geq \int_0^{\val} \hat{\alloc}(z) \dd z 
= \hat{\util}(\val).
\end{align*}
For any type with value $\val$ and participation cost $\cost$, 
since the utility for participation is higher in mechanism $\mech$, 
the agent has incentive to participate in mechanism $\mech$
only if he also has incentive to participate in mechanism $\widehat{\mech}$. 
Moreover, the payment difference is at most $\epsilon$.
Combining the observations, the expected revenue loss of mechanism $\mech$ is at most $\epsilon$.  
\end{proof}

\thmmenu*
\begin{proof}
Suppose the support of the cost distribution is 
$\{c_1,\dots,c_d\}$ 
where $c_j < c_{j+1}$ for any $1\leq j \leq d-1$.
By \Cref{lem:revelation}, it is without loss to consider a single allocation rule and 
a corresponding incentive compatible payment rule 
and let the buyer decide whether to participate or not
and, contingent on participation, 
choose her desired allocation and price. 
Suppose $\alloc^*$ is the optimal allocation function
and $\pay^*(\val) = \val \alloc^*(\val) - \int_0^{\val} \alloc^*(z) \dd z + p^*_0$
is the corresponding payment.  Recall that $\val_{x^*}(\cost)$ is the threshold value at which a buyer with participation cost $\cost$ will participate in the mechanism with allocation rule $x^*$.
Define $c_0\equiv -p^*_0$ and $\val_{x^*}(c_0) \equiv 0$.
Note that for any allocation $\alloc$ satisfying 
\begin{align}\label{eq:integration equiv}
\int_{\val_{x^*}(c_j)}^{\val_{x^*}(c_{j+1})} \alloc(z) \dd z
= \int_{\val_{x^*}(c_j)}^{\val_{x^*}(c_{j+1})} \alloc^*(z) \dd z
\end{align}
for all $j\in \{0,\dots,d-1\}$, 
the cutoff value for the buyer to participate the auction $\val_{x^*}(c_j)$ is not affected for any~$j\in \{1,\dots,d\}$. 
Thus it is sufficient to consider the optimization problem between $\val_{x^*}(c_j)$ 
and $\val_{x^*}(c_{j+1})$
respectively for any $j\in \{0,\dots,d-1\}$.
Note that the allocation~$\alloc$ must be non-decreasing between $\val_{x^*}(c_j)$ 
and $\val_{x^*}(c_{j+1})$, 
and satisfy the integration constraint in Equation~\eqref{eq:integration equiv}.
Thus the revenue maximization problem is reduced to maximizing the expected virtual value $\alloc_{\cost}(\val)\virtual_c(\val)$
for values between $\val_{x^*}(c_j)$ 
and $\val_{x^*}(c_{j+1})$, 
subject to the integration constraint for the allocation rule.
This is mathematically equivalent to solving a revenue maximization problem subject to a public budget constraint,
which is also virtual value maximization subject to an integration constraint on the allocation rule. 
In \citet{devanur2017optimal}, the authors have shown that the menu size of the optimal mechanism 
for this optimization problem is at most 2.
Thus we need at most $2d$ menu entries for $d$ separate programs to optimize the allocation below the cutoff $\val_{x^*}(c_d)$. 
Note that for optimizing the allocation above the cutoff $\val_{x^*}(c_d)$, 
the integration constraint is not required, 
and similar to \citet{myerson1981optimal}, a single menu entry is sufficient for the revenue maximization problem with linear buyers. 
Thus in total the menu size required is at most $2d+1$.
\end{proof}

Before the proof of \cref{thm:sample complexity}, 
we first show that the difference in optimal revenue is small when the estimation error on the discrete probability distribution is small. 
\begin{lemma}\label{thm:discrete prob}
Let $T\subseteq [0,1]\times[0,1]$ be a set with finite size, 
and let $\jointF_1, \jointF_2$ be two distributions 
supported on $T$. 
If $\max_{t\in T} 
\inftynorm{\jointF_1(t) - \jointF_2(t)} \leq\epsilon$,
then $|\Rev(\jointF_1;\mech) - \Rev(\jointF_2;\mech)| \leq |T|\cdot\epsilon$
for any mechanism $\mech$ with non-negative payment.
\end{lemma}
\begin{proof}
Again it is sufficient to show that $\Rev(\jointF_1;\mech) - \Rev(\jointF_2;\mech) \leq |T|\cdot\epsilon$, 
and the case for $\Rev(\jointF_1;\mech) - \Rev(\jointF_1;\mech) \leq |T|\cdot\epsilon$ 
holds symmetrically.
Note that 
\begin{align*}
\Rev(\jointF_1;\mech) - \Rev(\jointF_2;\mech)
= \sum_{t\in T} \Rev(t;\mech)(\jointF_1(t) - \jointF_2(t))
\leq |T|\cdot\epsilon
\end{align*}
since $\Rev(t;\mech)\leq 1$ for any $t\in T$ by individual rationality. 
Hence \Cref{thm:discrete prob} holds.
\end{proof}

\thmsample*
\begin{proof}
For any joint distribution $\jointF$, 
we construct the discrete distribution $\jointF'$ supported on the grid with increment~$\frac{\epsilon}{6}$ in support 
$[0,1] \times [0,1]$
by rounding down the value and rounding up the participation cost to the discretized points.
By \Cref{thm:discrete}, 
we have $\opt(\jointF') \geq \opt(\jointF) - \frac{\epsilon}{2}$.

Note that the size of the support of distribution $\jointF'$ is $36\epsilon^{-2}$.
We construct the distribution 
$\jointF\primed$ by
rounding down the value and rounding up the participation cost to the multiples of $\frac{\epsilon}{6}$
for any sample $t\sim \jointF$.\footnote{To clarify, $\jointF'$ rounds the true underlying distribution to the discrete support while $\jointF''$ rounds the empirical distribution to the discrete support.}
Note that it is easy to verify that $\jointF\primed$ is the empirical distribution for~$\jointF'$.
Moreover, for any $t$ in the support of distribution $\jointF'$, 
by Hoeffding's inequality (\cref{lem:Hoeffding}),
with $O(\epsilon^{-6}\log \epsilon^{-1})$ samples, 
$|\jointF'(t) - \jointF\primed(t)| \leq \frac{\epsilon^3}{288}$
with probability at least $1-\frac{\epsilon^3}{288}$.
By union bounds, 
with probability at least $1-\frac{\epsilon}{8}$ we have that 
$|\jointF'(t) - \jointF\primed(t)| \leq \frac{\epsilon^3}{288}$
for all $t$ in the support of distribution $\jointF'$.
Let $\mech$ be the optimal mechanism for empirical distribution $\jointF\primed$
and let $\mech'$ be the optimal mechanism for distribution $\jointF'$. 
By \cref{lem:p0=0}, both $\mech$ and $\mech'$ have non-negative payment, 
and 
\begin{align*}
\Rev(\jointF; \mech) 
&\geq \Rev(\jointF'; \mech)
\geq (1-\frac{\epsilon}{8})(\Rev(\jointF\primed; \mech) - \frac{\epsilon}{8})\\
&\geq \opt(\jointF\primed) - \frac{\epsilon}{4}
\geq \Rev(\jointF\primed; \mech') - \frac{\epsilon}{4}\\
&\geq (1-\frac{\epsilon}{8})(\Rev(\jointF'; \mech') - \frac{\epsilon}{8}) - \frac{\epsilon}{4}
\geq \opt(\jointF') - \frac{\epsilon}{2} 
\geq \opt(\jointF) - \epsilon.
\end{align*}
The first inequality holds since decreasing the value and increasing the participation cost weakly decreases the expected revenue of the mechanism. 
The second and the fifth inequalities hold by applying \cref{thm:discrete prob} 
and the observation that 
with probability at least $1-\frac{\epsilon}{8}$, we have that
$|\jointF'(t) - \jointF\primed(t)| \leq \frac{\epsilon^3}{288}$
for all $t$ in the support of distribution $\jointF'$.
\end{proof}

\section{Missing Proofs for Pricing in Single-Buyer Setting}
\label{apx:pricing single}
\lemmonotone*
\begin{proof}
If $f'(\val) \leq 0$, 
we have 
\begin{align*}
\frac{\partial f(\val)(\phi(\val)-\cost)}{\partial \val}
= \frac{\partial f(\val)\phi(\val)}{\partial \val} - f'(\val)\cdot \cost 
\geq 0.
\end{align*}
If $f'(\val) > 0$, 
we have 
\begin{align*}
\frac{\partial f(\val)(\phi(\val)-\cost)}{\partial \val}
= \frac{\partial f(\val)(\val-\cost) - (1-F(\val))}{\partial \val} 
= f'(\val)(\val-\cost) + 2f(\val) \geq 0
\end{align*}
where the inequality holds for $\val \geq \cost$. 
\end{proof}

\thmdmr*
\begin{proof}
Here we will prove the result for a slightly more general setting, 
where the participation costs can be correlated with the values, 
but the conditional value distribution~$\jointF_c$ has
identical and bounded support, and has decreasing marginal revenue for any participation cost~$\cost$. 

Denote the value upper bound by~$H < \infty$.
For any allocation rule $\alloc$ 
and associated payment rule $\pay$,
let $\bar{\val} = \sup_{\val\leq H} \{\alloc(\val) < 1\}$,
and $\mu = \int_{0}^{\bar{\val}} (1-\alloc(z)) \dd z$.
Note that $\mu$ is finite since $H < \infty$.
Let 
\begin{align}\label{eq:alloc}
\hat{\alloc}(\val) = \begin{cases}
1 & \val \geq \mu\\
0 & \val < \mu.
\end{cases}
\end{align}
\begin{figure}[t]
\begin{center}
\input{fig/DMR-x}
\vspace{-6pt}
\caption{\label{fig:DMR}
Given any allocation rule $\alloc(\val)$ (dashed line),
$\hat{x}$ (solid line) is the allocation rule for a posted price mechanism such that 
$\int_{0}^{\bar{v}}x(v) \,{\rm d}v
= \int_{0}^{\bar{v}}\hat{x}(v) \,{\rm d}v$.
The area of the shaded region is $\cost$.}
\end{center}
\vspace{-15pt}
\end{figure}
Allocation rule $\hat{\alloc}$ is illustrated in \Cref{fig:DMR}.
For any participation cost $\cost$, 
it is easy to verify that 
\begin{align*}
\hat{\alloc}_{\cost}(\val) = \begin{cases}
1 & \val \geq \val_{\cost}(\hat{\alloc})\\
0 & \val < \val_{\cost}(\hat{\alloc}).
\end{cases}
\end{align*}
where $\val_{\cost}(\hat{\alloc}) = \bar{\val} - \mu + \cost$.
Moreover, $\int_0^{\vupper} \alloc(z) \dd z = \int_0^{\vupper} \hat{\alloc}(z) \dd z$ 
and $\int_0^{\val} \alloc(z) \dd z \geq \int_0^{\val} \hat{\alloc}(z) \dd z$ for any $\val\in [0,\vupper]$. 
Thus by \Cref{thm:revenue equivalence}, 
for any $\cost\geq 0$,
the revenue of mechanism with allocation~$\alloc$ given participation cost $\cost$ is
\begin{align}
\rev(\alloc; \cost) 
&= 
\expect[\val\sim F]{\alloc_{\cost}(\val)\virtual_c(\val)}
- (1-\jointF_c(\val_{\alloc}(\cost)))\cdot \cost
= \int_{\val_{\alloc}(\cost)}^{\vupper} \jointf_{\cost}(\val)(\alloc_{\cost}(\val)\virtual_{\cost}(\val) - \cost) \dd\val\nonumber\\
&\leq \int_{\val_{\alloc}(\cost)}^{\vupper} \jointf_{\cost}(\val)\alloc_{\cost}(\val)(\virtual_{\cost}(\val) - \cost) \dd\val\nonumber\\
&= \jointf_{\cost}(\val)(\virtual_{\cost}(\val) - \cost)\int_{0}^{\val}\alloc_c(z) \dd z {\Bigg|}_{\val=\val_{\alloc}(\cost)}^{\vupper}
- \int_{\val_{\alloc}(\cost)}^{\vupper} \int_{0}^{\val}\alloc_{\cost}(z) \dd z
\dd[\jointf_{\cost}(\val)(\virtual_{\cost}(\val) - \cost)]
\nonumber\\
&\leq \jointf_{\cost}(\val)(\virtual_{\cost}(\val) - \cost)\int_{0}^{\val}\hat{\alloc}_c(z) \dd z {\Bigg|}_{\val=\val_{\hat{\alloc}}(\cost)}^{\vupper}
- \int_{\val_{\hat{\alloc}}(\cost)}^{\vupper} \int_{0}^{\val}\hat{\alloc}_{\cost}(z) \dd z
\dd[\jointf_{\cost}(\val)(\virtual_{\cost}(\val) - \cost)]
\nonumber\\
&= \int_{\val_{\hat{\alloc}}(\cost)}^{\vupper} \jointf_{\cost}(\val)\hat{\alloc}_{\cost}(\val)(\virtual_{\cost}(\val) - \cost) \dd\val
= \int_{\val_{\hat{\alloc}}(\cost)}^{\vupper} \jointf_{\cost}(\val)(\hat{\alloc}_{\cost}(\val)\virtual_{\cost}(\val) - \cost) \dd\val
= \rev(\hat{\alloc}; \cost).\label{eq:rev increase}
\end{align}
The first inequality holds because $\cost\geq 0$ and $\alloc_{\cost}(\val)\in[0,1]$ for any $\val$.
The third and the fourth equalities hold by integration by parts. 
The fifth equality holds because $\hat{\alloc}_{\cost}(\val)=1$ for any $\val\geq \val_{\hat{\alloc}}(\cost)$.
The second inequality follows from the combination of two observations: 
\begin{enumerate}
\item $\alloc_{\cost}(\val) = 0$ 
for any $\val \leq \val_{\alloc}(\cost)$
and $\hat{\alloc}_{\cost}(\val) = 0$ 
for any $\val \leq \val_{\hat{\alloc}}(\cost)$. 
This further implies that 
$\int_{0}^{\val_{\alloc}(\cost)}\alloc_{\cost}(z) \dd z = \int_{0}^{\val_{\hat{\alloc}}(\cost)}\hat{\alloc}_{\cost}(z) \dd z = 0$. 
Hence 
\begin{align*}
\jointf_{\cost}(\val)(\virtual_{\cost}(\val) - \cost)\int_{0}^{\val}\alloc_c(z) \dd z {\Bigg|}_{\val=\val_{\alloc}(\cost)}^{\vupper}
= \jointf_{\cost}(\val)(\virtual_{\cost}(\val) - \cost)\int_{0}^{\val}\hat{\alloc}_c(z) \dd z {\Bigg|}_{\val=\val_{\hat{\alloc}}(\cost)}^{\vupper}.
\end{align*}

\item By the construction of $\hat{\alloc}$, 
we have $\int_0^{\val} \alloc(z) \dd z \geq \int_0^{\val} \hat{\alloc}(z) \dd z \geq 0$ for any $\val\geq 0$,
which implies $\val_{\alloc}(\cost) \leq \val_{\hat{\alloc}}(\cost)$
for any $\cost$. 
Moreover, by \Cref{lem:monotone},
$\jointf_{\cost}(\val)(\virtual_{\cost}(\val) - \cost)$ 
is non-decreasing in $\val$ for any $\val \geq \val_{\hat{\alloc}}(\cost)\geq \cost$
and $\val \leq \vupper$.
Hence, we have 
\begin{align*}
&\int_{\val_{\alloc}(\cost)}^{\vupper} \int_{0}^{\val}\alloc_{\cost}(z) \dd z
\dd[\jointf_{\cost}(\val)(\virtual_{\cost}(\val) - \cost)]\\
\geq\,&
\int_{\val_{\hat{\alloc}}(\cost)}^{\vupper} \int_{0}^{\val}\alloc_{\cost}(z) \dd z
\dd[\jointf_{\cost}(\val)(\virtual_{\cost}(\val) - \cost)]\\
\geq\,&
\int_{\val_{\hat{\alloc}}(\cost)}^{\vupper} \int_{0}^{\val}\hat{\alloc}_{\cost}(z) \dd z
\dd[\jointf_{\cost}(\val)(\virtual_{\cost}(\val) - \cost)]
\end{align*}
where the first inequality holds since 
$\val_{\alloc}(\cost) \leq \val_{\hat{\alloc}}(\cost)$,
$\int_0^{\val} \alloc(z) \dd z \geq 0$
and $\jointf_{\cost}(\val)(\virtual_{\cost}(\val) - \cost)$ 
is non-decreasing in $\val$.
The second inequality holds since 
$\int_0^{\val} \alloc(z) \dd z \geq \int_0^{\val} \hat{\alloc}(z) \dd z$
and $\jointf_{\cost}(\val)(\virtual_{\cost}(\val) - \cost)$ 
is non-decreasing in $\val$.
\end{enumerate}
Taking expectation over $\cost$, \cref{thm:dmr} holds. 
\end{proof}

\thmmhr*
\begin{proof}
For any allocation rule $\alloc(\val)$,
let $\cost' = \inf_{\cost\geq 0} \{\virtual(\val_{\alloc}(\cost)) - \cost \geq 0\}$
and let $\val' = \val_{\alloc}(\cost')$. 
We first focus on the case that both $\cost'$ and $\val'$ are finite, 
i.e., there exists $\cost\geq 0$
such that $\virtual(\val_{\alloc}(\cost)) - \cost \geq 0$.\footnote{Since the value is independent of the participation cost, 
the virtual value function is invariant of the participation cost of the buyer. }
For any allocation rule $\alloc(\val)$,
let 
\begin{align*}
\hat{\alloc}(\val) = \begin{cases}
1 & \val \geq \val'\\
\alloc(\val) & \val < \val'.
\end{cases}
\end{align*}
\begin{figure}[t]
\begin{center}
\input{fig/MHR-x}
\vspace{-6pt}
\caption{\label{fig:MHR}
Given any allocation rule $\alloc(\val)$ (black dashed line),
$\hat{x}$ (red solid line) is the allocation rule 
that increases the allocation to $1$ for value aboves $v'$
and $\hat{x}'$ (blue dashed line) is the allocation rule that corresponds to posting price $\mu$.
The two shaded regions in the figure have the same area.}
\end{center}
\vspace{-15pt}
\end{figure}%
Allocation $\hat{\alloc}$ is illustrated in \Cref{fig:MHR}.
By \cref{thm:revenue equivalence},
since we set $p_0=0$,
the expected revenue of any mechanism with allocation rule $\alloc$ from the buyer with participation cost~$\cost\geq 0$ is 
\begin{align*}
\expect[\val\sim F]{\alloc_{\cost}(\val)\virtual(\val)} - (1-F(\val_{\alloc}(\cost)))\cdot \cost
= \int_{\val_{\alloc}(\cost)}^{\infty} f(\val)(\alloc_{\cost}(\val)\virtual(\val) - \cost) \dd\val.
\end{align*}
We separate the discussion into two cases. 
For any participation cost $c$ such that $\val_{\alloc}(\cost) \leq \val'$, 
we have $\val_{\alloc}(\cost) = \val_{\hat{\alloc}}(\cost)$. 
In this case, since $\phi(\val) \geq 0$ for any $\val \geq \val'$,
allocation $\hat{\alloc}$ only increases the allocation to types with positive virtual value~$\virtual(\val)$ compared to $\alloc$, 
which increases the expected revenue. 
For any participation cost $c$ such that $\val_{\alloc}(\cost) > \val'$,
we have $\val_{\alloc}(\cost) \geq \val_{\hat{\alloc}}(\cost) > \val'\geq 0$. 
In this case, for any 
$\tilde{\val} \in [\val_{\hat{\alloc}}(\cost), \val_{\alloc}(\cost)]$,
we have 
\begin{align*}
\virtual(\tilde{\val}) - \cost 
\geq \virtual(\val_{\hat{\alloc}}(\cost)) - \cost
\geq \virtual(\val') - c'
= 0,
\end{align*}
where the first inequality holds since 
$\tilde{\val}\geq \val_{\hat{\alloc}}(\cost)$.
The second inequality holds because 
$\virtual(\val) - \val$ is non-decreasing in $\val$
due to the MHR assumption, 
and $\val_{\hat{\alloc}}(\cost) \geq v'$.
Thus increasing the allocation to 1 for values between 
$[\val_{\hat{\alloc}}(\cost), \val_{\alloc}(\cost)]$ 
only weakly increases the revenue. 
Therefore, 
the mechanism with allocation rule $\hat{\alloc}$ generates weakly higher revenue compared to $\alloc$.

Let $\bar{\val} = \sup_{\val} 
\{\hat{\alloc}(\val) < 1 \text{ and } F(\val) < 1\}$
and $\mu = \int_{0}^{\bar{\val}} (1-\alloc(z)) \dd z$.
We claim that $\mu$ is finite. 
In the case that $\val'$ is finite or the maximum value in the support is $H < \infty$, 
we have that $\bar{\val} \leq \min\{\val', H\}$ is finite,
and hence $\mu$ is finite. 
In the case that $\val'$ is infinite
and $\bar{\val}$ is infinite, 
let $m$ be the value such that $\virtual(m)=0$.\footnote{For an MHR distribution, 
the value $m$ with virtual value zero is always finite \citep{hartline2008optimal}. } 
Suppose $\mu$ is infinite. 
Then for sufficiently large $c$, we have 
\begin{align*}
\virtual(\val_{\alloc}(\cost)) - \cost 
\geq \val_{\alloc}(\cost) - m - \cost 
= c + \int_0^{\val_{\alloc}(\cost)}(1-\alloc(z)) \dd z - m - \cost 
> 0
\end{align*}
where the first inequality holds because $\phi'(v) \geq 1$ since $F$ is MHR, 
and the equality holds by the definition of $\val_{\alloc}(\cost)$.
The last inequality holds since because $\mu$ is infinite,
for any $m$
there exist a sufficiently large $c$ such that 
$\int_0^{\val_{\alloc}(\cost)}(1-\alloc(z)) \dd z > m$.
Note that this contradicts to the condition that $\val'$ is infinite, 
and hence $\mu$ must be finite.
Let 
\begin{align*}
\tilde{\alloc}(\val) = \begin{cases}
1 & \val \geq \mu\\
0 & \val < \mu.
\end{cases}
\end{align*}
This is well defined since $\mu$ is finite. 
Allocation $\tilde{\alloc}$ is illustrated in \Cref{fig:MHR}.
For any participation cost~$\cost$ such that $\val_{\hat{\alloc}}(\cost) \geq \bar{\val}$, 
we have that 
$\val_{\hat{\alloc}}(\cost) = \val_{\tilde{\alloc}}(\cost)$
and the revenues are the same for both mechanisms. 
For any participation cost~$\cost$ such that $\val_{\hat{\alloc}}(\cost) < \bar{\val}$, 
we have that $\cost < \cost'$
since $\val_{\tilde{\alloc}}(\cost') = \val_{\alloc}(\cost') = \val'$
and $\bar{\val} \leq \val'$.
In this case, let $\hat{\phi}$ be the virtual value function such that 
$\hat{\phi}(\val) = \phi(\val)-c$ if and only if $\phi(\val) - \cost < 0$, 
and $\hat{\phi}(\val) = 0$ otherwise. 
By \cref{lem:regular imply dmr for small v}, 
it is easy to verify that the valuation distribution with virtual function $\hat{\phi}$ has decreasing marginal revenue. 
we have that 
\begin{align*}
&\rev(\hat{\alloc}; \cost) 
= \int_{\val_{\hat{\alloc}}(\cost)}^{\bar{\val}} f(\val)(\hat{\alloc}_{\cost}(\val)\virtual(\val) - \cost) \dd\val
+\int_{\bar{\val}}^{\infty} f(\val)(\hat{\alloc}_{\cost}(\val)\virtual(\val) - \cost) \dd\val\\
&\leq \int_{\val_{\hat{\alloc}}(\cost)}^{\bar{\val}} f(\val)\hat{\alloc}_{\cost}(\val)(\virtual(\val) - \cost) \dd\val
+\int_{\bar{\val}}^{\infty} f(\val)(\tilde{\alloc}_{\cost}(\val)\virtual(\val) - \cost) \dd\val\\
&= \int_{\val_{\hat{\alloc}}(\cost)}^{\bar{\val}} f(\val)\hat{\alloc}_{\cost}(\val)\hat{\virtual}(\val) \dd\val
+ \int_{\val_{\hat{\alloc}}(\cost)}^{\bar{\val}} f(\val)(\hat{\alloc}_{\cost}(\val)(\virtual(\val) - \cost -\hat{\virtual}(\val))) \dd\val
+\int_{\bar{\val}}^{\infty} f(\val)(\tilde{\alloc}_{\cost}(\val)\virtual(\val) - \cost) \dd\val\\
&= \int_{\val_{\hat{\alloc}}(\cost)}^{\bar{\val}} f(\val)\hat{\alloc}_{\cost}(\val)\hat{\virtual}(\val) \dd\val
+ \int_{\val_{\tilde{\alloc}}(\cost)}^{\bar{\val}} f(\val)(\tilde{\alloc}_{\cost}(\val)(\virtual(\val) - \cost -\hat{\virtual}(\val))) \dd\val
+\int_{\bar{\val}}^{\infty} f(\val)(\tilde{\alloc}_{\cost}(\val)\virtual(\val) - \cost) \dd\val\\
&\leq \int_{\val_{\tilde{\alloc}}(\cost)}^{\bar{\val}} f(\val)\tilde{\alloc}_{\cost}(\val)\hat{\virtual}(\val) \dd\val
+ \int_{\val_{\tilde{\alloc}}(\cost)}^{\bar{\val}} f(\val)(\tilde{\alloc}_{\cost}(\val)(\virtual(\val) - \cost -\hat{\virtual}(\val))) \dd\val
+\int_{\bar{\val}}^{\infty} f(\val)(\tilde{\alloc}_{\cost}(\val)\virtual(\val) - \cost) \dd\val\\
&= \rev(\tilde{\alloc}; \cost).
\end{align*}
The first inequality holds because the allocation 
$\hat{\alloc}_{\cost}(\val)\in [0,1]$ for any $\val$
and $\hat{\alloc}_{\cost}(\val) = \tilde{\alloc}_{\cost}(\val)$
for any $\val\geq \bar{\val}$. 
The third equality holds since $\virtual(\val) - \cost -\hat{\virtual}(\val)$ is non-negative, 
$\val_{\tilde{\alloc}}(\cost) \geq \val_{\hat{\alloc}}(\cost)$,
and $\virtual(\val) - \cost -\hat{\virtual}(\val) = 0$
for any value $\val\leq \val_{\tilde{\alloc}}(\cost)$. 
The last statement holds because under allocation~$\tilde{\alloc}$, 
we have $\cost'-\cost = \val_{\tilde{\alloc}}(\cost') - \val_{\tilde{\alloc}}(\cost)$
and hence $\virtual(\val) - \cost \leq \virtual(\tilde{\alloc}(\cost)) - \cost
\leq \virtual(\val_{\tilde{\alloc}}(\cost')) - \cost' = 0$
for any value $\val\leq \tilde{\alloc}(\cost)$
since the valuation distribution is MHR. 
The second inequality holds by applying Inequality \eqref{eq:rev increase}
since $\hat{\virtual}$ can be viewed as the virtual value function for distribution with decreasing marginal revenue, 
and the allocation rules are converted through the same format. 
Taking expectation over~$\cost$, \cref{thm:mhr} holds. 
\end{proof}

\thmCorrelate*
\begin{proof}
For any mechanism $\mech$, 
let $\val_0$ be the minimum value of any agent that participates the auction. 
Then agents of this type will be indifferent between participating and not participating, and hence have utility $u_0 = c_{v_0}$. 
Let $\alloc_0 = \frac{u_0}{v_0}$.

Now consider the following alternative revenue maximization problem for a single buyer.
This alternative problem will have the same distribution over values as the original setting, but all participation costs are set equal to~$0$.  Instead of participation costs, we impose three constraints on the class of mechanisms that can be used. 
First, the seller is constrained to only sell the item to agents with value above $\val_0$ subject to the incentive constraint. 
Second, the allocation returned by the mechanism is constrained to be at least~$\alloc_0$.
Third, the utility of an agent with value $\val_0$ is constrained to be exactly $u_0$. 

What is the revenue-optimal mechanism for this alternative problem?  Since there are no participation costs, 
a direct implication of \citet{myerson1981optimal}
is that the revenue optimal mechanism $\mech'$ in this alternative setting, subject to the allocation and utility constraints, is a step function. 
Specifically, since the minimum allocation is at least $\alloc_0$, 
this corresponds to a mechanism with menu size $2$, 
where one of the menu entries is $(x_0, 0)$
and the other is $(1, p)$ with $p\geq v_0-u_0$.
The utility function of mechanism $\mech'$ is illustrated in \Cref{fig:concave}\begin{figure}[t]
\begin{center}
\input{fig/concave}
\vspace{-6pt}
\caption{\label{fig:concave}
The black (respectively red) solid curve is the utility function of the agent (without paying the cost) given any mechanism $\mech$ (respectively $\mech'$),
and the black dashed curve is the participation cost function $\cost_{\val}$.
All three curves intersect at point $(\val_0, u_0)$.
}
\end{center}
\vspace{-15pt}
\end{figure}
 as the red solid line.
Note that since mechanism $\mech$ is also a feasible mechanism for this alternative revenue maximization problem, 
we must have $\Rev(\mech) \leq \Rev(\mech')$ (where revenue is calculated with respect to the new setting).

Now consider any mechanism that is feasible for the new setting.
Since the participation cost $\cost_{\val}$ is concave in $\val$, 
the utility of an agent participating in the mechanism with value $\val \geq \val_0$ is at least $\val\cdot\alloc_0\geq \cost_{\val}$.
Therefore, returning to the original problem formulation with participation costs, any agent with value $\val\geq \val_0$ (and hence participation cost $\cost_{\val}$) will choose participate in the auction in the original setting.  This means that for both $\mech$ and $\mech'$ (both of which are feasible in the new setting), the revenue remains unchanged when executing the mechanism in the original setting.
This implies that in the original setting, 
$\Rev(\mech) \leq \Rev(\mech')$.
Finally, since the payment for menu entry $(x_0, 0)$ is~$0$, removing this entry from the mechanism $\mech'$ weakly improves the expected revenue.  The resulting mechanism is a posted-price mechanism.
Hence, if the participation cost $\cost_{\val}$ is non-negative and concave in~$\val$, 
posted pricing is a revenue optimal mechanism.
\end{proof}

\section{Multi-buyer Setting}
\label{apx:multi-buyer}
\subsection{Non-Convexity}
\label{sub:nonconvex}
When there are multiple buyers, a common approach in mechanism design is to represent the mechanism by interim allocations and payments. 
In particular, let 
\begin{align*}
\alloc_i(\val_i, \cost_i) \triangleq \expect[\val_{-i},\cost_{-i}]{\alloc_i(\val, \cost)}
\text{ and }
\pay_i(\val_i, \cost_i) \triangleq \expect[\val_{-i},\cost_{-i}]{\pay_i(\val, \cost)}.
\end{align*}
In \citet{border1991implementation,che2013generalized}, 
the authors provide sufficient and necessary conditions on the set of interim allocations that are implementable. 
Then the revenue maximization problem can be formalized as the following optimization program. 
\begin{align*}
\max_{\alloc,\pay} \quad& \expect[\val,\cost]{\sum_i\pay_i(\val_i, \cost_i)}\\
\text{s.t.} \quad& 
\alloc_i(\val_i, \cost_i) \cdot \val_i - \pay_i(\val_i, \cost_i)
\geq \alloc_{i}(\val'_i, \cost'_i) \cdot \val_i - \pay_{i}(\val'_i, \cost'_i), 
\qquad\forall i,\val,\val',\cost,\cost'\\
& \alloc_i(\val_i, \cost_i) = \pay_i(\val_i, \cost_i) = 0
\text{ or } \alloc_i(\val_i, \cost_i) \cdot \val_i - \pay_i(\val_i, \cost_i)\geq \cost, 
\qquad\forall i,\val,\cost\\
&x \text{ is implementable according to \citet{border1991implementation}.}
\end{align*}
It is easy to see that the above optimization problem is not a convex program. 
A natural conjecture is that whether one could reformulate the problem such that it can be represented as a convex program. 
In the following example, we show that this is not the case. 

\begin{example}
There are two identical buyers. For each buyer, with probability 1, his value is~2 and his participation cost is 1. 
In this case, the optimal mechanism is to sell the item to buyer~1 with price 1 (or sell the item to buyer 2 with price 1) with expected revenue 1. 
Note that this mechanism is asymmetric. 
In fact, for any symmetric mechanism, the probability the item is sold to each buyer is at most $\sfrac{1}{2}$, 
and to satisfy the individual rationality constraint, 
the payment from each buyer is non-positive. 
Thus the revenue from the symmetric mechanism is at most $0$, 
which is smaller than the optimal revenue. 
\end{example}

Note that if the environment is symmetric and the problem can be represented as a convex program, 
there must exist a symmetric mechanism that is optimal.
However, the above example illustrates that the optimal mechanism is not symmetric for the multi-buyer setting in symmetric environments,
which rules out the possibility of restructuring the optimization program into a convex one. 
This observation illustrates a distinction between our model and other inter-dimensional problems, 
where the optimization program for the latter cases are often linear programs. 
Note that in general for non-convex programs, 
we cannot hope to derive succinct closed-form solutions or compute it in polynomial time. 
However, for the problem of revenue maximization for buyers with participation costs, 
we propose simple mechanisms that are approximately optimal under reasonable assumptions on the distributions. 

\subsection{Missing Proofs for Computation}
\label{sub:multi-proof}
Before the proof of \cref{lem:compute_approx}, we introduce the following two lemmas for bounding the discretization errors. 

\begin{restatable}{lemma}{lemfixed}\label{lem:ex-ante fixed}
For any product distribution $\jointF=\times_{i=n}\jointF_i$ with $\jointF_i$ supported on $[0,1]^2$, 
for any $\epsilon\in (0,1)$ and the corresponding discretized distribution $\jointF\dprimed$, 
for any profile of ex-ante probabilities $\{\quant_i\}_{i\in[n]}$ with $\sum_i\quant_i\leq 1$, 
there exists another profile of ex-ante probabilities $\{\quant\dprimed_i\}_{i\in[n]}$ in grid $\grid_0$
with $\sum_i\quant\dprimed_i\leq 1$
such that 
\begin{align*}
\sum_{i\in[n]}\optd_{\quant\dprimed_i}(\jointF\dprimed_i)
\geq \sum_{i\in[n]}\opt_{\quant_i}(\jointF_i) - 5n\epsilon.
\end{align*}
\end{restatable}
\begin{proof}
Recall that $\jointF'$ is the distribution that rounding the values down to the discretization grid~$\grid_0$. 
Let $\mech,\mech'$ and $\mech\dprimed$ be the optimal mechanisms with allocation and payment rule 
$(\alloc,\pay), (\alloc',\pay')$ and $(\alloc\dprimed,\pay\dprimed)$
under distributions $\jointF_i,\jointF'_i$ and $\jointF\dprimed_i$. 

First note that $\opt_{\quant_i}(\jointF\dprimed_i) \geq \opt_{\quant_i}(\jointF'_i)$
since $\jointF\dprimed_i$ is constructed by decreasing the participation cost compared to $\jointF'_i$.
Next we bound the expected revenue loss between $\jointF_{i}$ and $\jointF'_{i}$
given any ex-ante constraint $\quant_i$. 
Consider a random boosting $z_j$ drawn from the distribution over value difference within in interval $[j\cdot \epsilon, (j+1)\epsilon)$ between distributions $\jointF$ and $\jointF'$
for all $j\leq \frac{1}{\epsilon}$.
Let $\mech_0$ be the mechanism that announces the realization of $z_j$ for all $j$, 
and then offer allocation $\alloc(j\cdot\epsilon^2+z_j)$ with payment $\pay(j\cdot\epsilon^2 + z_j) - \alloc(j\cdot\epsilon^2+z_j)\cdot z_i$
if the buyer reports value $j\cdot\epsilon^2$.%
\footnote{Mechanism $\mech$ may be a distribution over mechanisms, 
and in this case, mechanism $\mech_0$ is also a distribution over mechanisms by applying this procedure for each realization of the mechanisms in $\mech$.}
It is easy to verify that all values in the support of $\jointF'_i$ has incentives to report truthful in mechanism $\mech_0$, 
and the expected allocation of mechanism $\mech_0$ given $\jointF'_i$ coincide with the expected allocation of $\mech$ given $\jointF_i$. 
Thus,
\begin{align}
\opt_{\quant_i}(\jointF\dprimed_i) &\geq\opt_{\quant_i}(\jointF'_i) 
\geq \Rev(\jointF'_i; \mech_0)
\geq \Rev(\jointF_i; \mech) - \max_{j,z_j}\alloc(j\cdot\epsilon^2+z_j)\cdot z_j\nonumber\\
&\geq \Rev(\jointF_i; \mech) - \epsilon
= \opt_{\quant_i}(\jointF_i) - \epsilon.\label{eq:discrete dist}
\end{align}
Now consider another mechanism $\mech_1$ with parameter $j^*$ such that for any value below $j^*\cdot\epsilon^2$, allocation $\alloc\dprimed$ is rounded down to the multiples of $\epsilon$, 
and for any value above $j^*\cdot\epsilon^2$, allocation $\alloc\dprimed$ is round up to the multiples of $\epsilon$.
Allocation for value $j^*\cdot\epsilon^2$ is rounded randomly. 
Parameter~$j^*$ and the rounding probability at $j^*\cdot\epsilon^2$ is chosen such that the ex-ante feasibility is preserved. 
Note that in $\mech_1$, the realization of the random rounding is disclosed to the buyer. 
Let $\optt_{\quant_i}(\jointF\dprimed_i)$ be the optimal revenue for distribution $\jointF\dprimed_i$ given ex-ante constraint $\quant_i$
only using randomize over mechanisms with allocations on the discretized grid $\grid_0$. 
We have 
\begin{equation}\label{eq:discrete x}
\optt_{\quant_i}(\jointF\dprimed_i) \geq \Rev(\jointF\dprimed_i; \mech_1)
\geq \Rev(\jointF\dprimed_i; \mech\dprimed) - \epsilon
= \opt_{\quant_i}(\jointF\dprimed_i)  - \epsilon.
\end{equation}
Given a profile of ex-ante constraints
$\{\quant_i\}_{i\in[n]}$,
there exists a profile over random ex-ante constraints $\{\quant'_i\}_{i\in[n]}$
with corresponding fixed mechanisms $\mech_{i,\quant'_i}$ that are $\quant'_i$ feasible
such that $\expect[]{\quant'_i} = \quant_i$
and 
$\expect[\quant'_i]{\Rev(\jointF\dprimed_i;\mech_{i,\quant'_i})} = \optt_{\quant_i}(\jointF\dprimed_i)$.
This is done by essentially examining the ex-ante allocation probability of each realized mechanism for each buyer~$i$, 
and rename that realized ex-ante allocation probability as variable $\quant'$. 
Moreover, consider another profile of random ex-ante probabilities $\{\tilde{\quant}_i\}_{i\in[n]}$ by rounding each realization of $\quant'_i$ to the grid~$\grid_0$. 
We have $\sum_{i\in[n]} \expect[]{\tilde{\quant}_i} \leq \sum_{i\in[n]} \expect[]{\quant'_i} \leq 1$.
Note that a feasible mechanism given ex-ante constraint $\tilde{\quant}_i$
is to offer the mechanism $\mech_{i,\quant'_i}$ with probability $\frac{\tilde{\quant}_i}{\quant''_i}$
and offer the mechanism with constant zero allocation and payment
with probability $1-\frac{\tilde{\quant}_i}{\quant''_i}$
where $\quant''_i$ is rounding $\quant_i$ up to the multiples of $\epsilon$.
The buyer is informed about which mechanism is offered before participation. 
Note that in this construction, mechanism $\mech_{i,\quant'_i}$
is also $\quant''_i$ feasible. 
Therefore, given this randomized mechanism, the ex-ante sale probability is $\tilde{\quant}_i$ 
and the expected revenue is 
\begin{align}
\expect[\tilde{\quant}_i]{\optd_{\tilde{\quant}_i}(\jointF\dprimed_i)} 
&\geq \expect[\tilde{\quant}_i,\quant'_i,\quant''_i]{\frac{\tilde{\quant}_i}{\quant''_i}\cdot \Rev(\jointF\dprimed_i;\mech_{i,\quant'_i}) }\nonumber
\geq \expect[\quant'_i]{\Rev(\jointF\dprimed_i;\mech_{i,\quant'_i}) - 2\epsilon\cdot \Rev(\jointF\dprimed_i;\mech_{i,\quant'_i})}\\
&\geq \expect[\quant'_i]{\Rev(\jointF\dprimed_i;\mech_{i,\quant'_i})} - 2\epsilon 
= \optt_{\quant_i}(\jointF\dprimed_i) - 2\epsilon \label{eq:discrete q}
\end{align}
where the second inequality holds since $|\quant''_i-\tilde{\quant}_i|\leq 2\epsilon$
and the third inequality holds since $\Rev(\jointF\dprimed_i;\mech_{i,\quant'_i})\leq 1$.
Note that given the distribution over ex-ante probabilities $\{\tilde{\quant}_i\}_{i\in[n]}$, 
to improve the sum of ex-ante revenue, 
we can greedily select a deterministic profile of ex-ante probabilities $\{\tilde{\quant}\dprimed_i\}_{i\in[n]}$
by ranking the realizations according to the ratio of 
realized $\frac{\optd_{\tilde{\quant}_i}(\jointF\dprimed_i)}{\tilde{\quant}_i}$, 
with the exception that there may exist one buyer $i^*$ such that $\tilde{\quant}\dprimed_i$ is selected randomly over two possible realizations. 
Note that $\tilde{\quant}\dprimed_i\in\grid_0$ for any $i\neq i^*$
since $\tilde{\quant}_i$ only randomize over ex-ante probabilities in grid $\grid_0$.
Finally, by setting $\quant\dprimed_i = \tilde{\quant}\dprimed_i$ for any $i\neq i^*$, 
and letting $\quant\dprimed_{i^*}$ be the expected value of $\tilde{\quant}\dprimed_i$ round down to multiples of $\epsilon$, 
we have 
\begin{align*}
\sum_{i\in [n]} \optd_{\quant\dprimed_i}(\jointF\dprimed_i) 
&\geq \sum_{i\in [n]} \expect[\tilde{\quant}\dprimed_{i^*}]{\optd_{\tilde{\quant}\dprimed_i}(\jointF\dprimed_i) } - \epsilon
\geq \sum_{i\in [n]}\expect[\tilde{\quant}_i]{\optd_{\tilde{\quant}_i}(\jointF\dprimed_i)} - \epsilon\\
&\geq \sum_{i\in [n]}\optt_{\quant_i}(\jointF\dprimed_i) - 3n\epsilon
\geq \sum_{i\in [n]} \opt_{\quant_i}(\jointF\dprimed_i) -4n\epsilon
\geq \sum_{i\in [n]} \opt_{\quant_i}(\jointF_i) 
-5n\epsilon
\end{align*}
where the last three inequalities are implied by inequalities 
\eqref{eq:discrete dist}, \eqref{eq:discrete x} and \eqref{eq:discrete q}.
\end{proof}

\begin{restatable}{lemma}{lempolysingle}\label{lem:poly_single}
In the single-buyer setting, for any distribution $\jointF$ supported on $\grid_0\times\grid_1$ and any $\quant\in [0,1]$, 
there exists an algorithm with running time $\poly{\frac{1}{\epsilon}}$ that computes the mechanism with ex-ante sale probability at most $\quant$
that optimizes $\optd_{\quant}(\jointF\dprimed_i)$.
\end{restatable}
\begin{proof}
We first use dynamic program to compute the optimal revenue from fixed mechanisms for $\quant$ in grid $\grid_0$. 
For $i,j\leq \frac{1}{\epsilon}, k\leq \frac{1}{\epsilon^2}$ and $s\leq \frac{1}{\epsilon^3}$, let $R(i,j,k,s)$ be the optimal revenue from types with value below $i\cdot\epsilon$
when the allocation and utility of value $i\cdot\epsilon$ are $j\cdot \epsilon$ and $k\cdot \epsilon^2$ respectively, 
and the total ex-ante allocation probability for types at most $i\cdot\epsilon$ is at most $s\cdot\epsilon^3$.
The optimal revenue from ex-ante constraint $\quant$
is determined by the entry 
$\max_{j,k}R(\frac{1}{\epsilon},j,k,\frac{\quant}{\epsilon^3})$.

To simplify notation, let $p_i$ be the integer such that the probability of value $i\cdot\epsilon$ in distribution $\jointF$
is $p_i\cdot\epsilon$, 
and $z_{ij}$ be the integer such that conditional on value, 
the probability such that the participation cost is at most $j\cdot\epsilon^2$ is $z_{ij}\cdot\epsilon$.
We initialize the matrix by $R(1,j,k,s) = 0$ for any $k\leq j$
and $s\geq p_i\cdot j\cdot z_{1k}$,
and $R(1,j,k,s) = -\infty$ otherwise. 
For any $i\geq 2$, we have 
\begin{align*}
R(i,j,k,s) = \max_{j'\leq j} R(i-1,j',k-j,s-p_i\cdot j\cdot z_{ik}) 
+ q_i\cdot \jointF'_{i}(k\cdot \epsilon^2) \cdot 
(i\cdot j-k)\cdot \epsilon^2.
\end{align*}
To interpret this expression, $R(i-1,j',k-j,s-p_i\cdot j\cdot z_{ik})$ is the revenue from types with values at most $(i-1)\epsilon$. 
Note that the expected allocation from value $(i-1)\epsilon$ is $p_i\cdot j\cdot z_{ik}$
the the utility is $k\cdot \epsilon^2$, 
and hence the total ex-ante allocation from values at most $(i-1)\cdot\epsilon$
cannot exceed $s-p_i\cdot j\cdot z_{ik}$.
Moreover, the allocation of type with value $(i-1)\epsilon$ should not exceed the allocation of type with value $i\cdot\epsilon$, 
and hence the choice of $j'$ is at most $j$.
Finally, when the allocation and utility of type with value $i\cdot \epsilon$ are $j\cdot \epsilon$
and $k\cdot \epsilon^2$ respectively, 
by \cref{lem:revelation}, the utility of type with value $(i-1) \epsilon$ is exactly $(k-j)\epsilon^2$.
This utility is independent of the choice of $j'$. 

The second term is the expected revenue from types with value at most $i\cdot\epsilon$.
Note that 
$q_i\cdot \jointF'_{i}(k\cdot \epsilon^2)$ is the total probability of those types that have incentives to participate in the auction, 
and $(i\cdot j-k)\cdot \epsilon^2$ is the payment of those types, 
which is derived by subtracting the utility $k\cdot \epsilon^2$ from the expected value $i\cdot j\cdot \epsilon^2$,
if they choose to participate. 

Finally, given any $\quant\in[0,1]$, it is sufficient to brute-force search for all pairs of ex-ante probabilities in grid $\grid_0$
such that their convex combination coincides with $\quant$.
This operation takes at most $\frac{1}{\epsilon^2}$ comparisons.
\end{proof}

\lemAproxComp*
\begin{proof}
First we compute the mechanism $\mech_{i,\quant\dprimed_i}$ that generates revenue $\optd_{\quant\dprimed_i}(\jointF\dprimed_i)$ for any buyer $i$
and~$\quant\dprimed_i$ in grid $\grid_0$. 
By Lemma~\ref{lem:poly_single} this requires $\frac{n}{\hat{\epsilon}} \cdot \poly{\frac{1}{\hat{\epsilon}}}$ computations. 
For $i\in[n]$ and $j\in[\lfloor\frac{1}{\hat{\epsilon}}\rfloor]$,
let $\rev(i,j)$ be the optimal revenue from the first $i$ buyers with total allocation probability $j\cdot\hat{\epsilon}$. 
We initialize by $\rev(1,j) = \optd_{j\cdot\hat{\epsilon}}(\jointF\dprimed_1)$ for any $j\in[\lfloor\frac{1}{\hat{\epsilon}}\rfloor]$.
We set 
\begin{align*}
R(i,j) = \max_{j'\leq j} R(i-1,j') +
\optd_{(j-j')\cdot\hat{\epsilon}}(\jointF\dprimed_i), 
\quad \forall i\geq 2, j\in[\lfloor\textstyle{\frac{1}{\hat{\epsilon}}}\rfloor].
\end{align*}
The dynamic program takes at most $\frac{1}{\hat{\epsilon}^3}$ operations. 
Therefore, we compute the profile of ex-ante probabilities $\{\quant\dprimed_i\}_{i\in[n]}$ in time $\poly{\frac{n}{\hat{\epsilon}}}$.

Note that directly running mechanism $\mech_{i,\quant\dprimed_i}$ 
on the actual distribution $\jointF_i$ for buyer $i$, rather than the rounded distribution $\jointF\dprimed_i$,
may violate the desired ex-ante constraint $\quant\dprimed_i$.  This is because the fraction of agents who opt out of the mechanism may differ between $\quant\dprimed_i$ and $\jointF_i$, since our rounding of values can change the utility enjoyed by the agent by as much as $\hat{\epsilon}$.
However, since the density of the cost distribution is at most $\eta$, this change in utility can influence the probability the item is sold by at most $\eta \hat{\epsilon}$, and hence
the ex-ante probability the item is sold given $\jointF_i$ 
is at most 
$\quant\dprimed_i+\eta\hat{\epsilon}$.
Moreover, the revenue given $\jointF_i$ is 
\begin{equation*}
\Rev(\jointF_i; \mech_{i,\quant\dprimed_i}) \geq \optd_{\quant\dprimed_i}(\jointF\dprimed_i) - \hat{\epsilon}.\qedhere
\end{equation*}
\end{proof}

\subsection{Missing Proofs for Approximations of Pricing}
\label{sub:multi-pricing}
\lemDmrSupply*
\begin{proof}
For any mechanism with allocation rule $\alloc$ and sale probability $q\leq 1$, 
construct the posted price mechanism $\hat{\alloc}$ as in Equation \eqref{eq:alloc}. 
By \Cref{thm:dmr}, the revenue of the mechanism with allocation $\hat{\alloc}$ is weakly higher than $\alloc$. 
Next we consider two cases. 

If the probability the item is sold given allocation rule $\hat{\alloc}$ is at most $q$, 
then $\hat{\alloc}$ is a feasible posted price mechanism with higher revenue, 
which implies that \Cref{lem:dmr supply} holds. 

If the probability the item is sold given allocation rule $\hat{\alloc}$ is higher than $q$, 
then there exists an allocation rule $\tilde{\alloc}$ that posts (a possibly randomized) price higher than $\hat{\alloc}$
that sells the item with probability~$q$.
To prove \Cref{lem:dmr supply}, 
it is sufficient to show that the per-unit price (expected payment divided by expected allocation) 
charged in mechanism with allocation rule $\tilde{\alloc}$ is always higher than~$\alloc$
for any type of the buyer, 
since mechanism with allocation rule $\tilde{\alloc}$ sells the item with probability exactly $q$, 
while mechanism with allocation rule $\alloc$
sells the item with probability at most $q$. 

Note that by the definition of allocation rule $\tilde{\alloc}$, 
the per-unit price charged under allocation rule $\tilde{\alloc}$ is always higher than $\hat{\alloc}$, 
where the latter equals
$\int_0^{\bar{\val}}(1-\alloc(z)) \dd z$. 
For allocation rule $\alloc$, 
if a type with value $\val \leq \bar{\val}$ participates the auction, 
the per-unit price for this type is 
\begin{align*}
\frac{\val \alloc(\val) - \int_0^\val \alloc(z) \dd z}{\alloc(\val)}
\leq \val - \int_0^\val \alloc(z) \dd z
\leq \int_0^{\bar{\val}}(1-\alloc(z)) \dd z.
\end{align*}
Finally, for a type with value $\val > \bar{\val}$ participating the auction, 
the per-unit price for this type equals that for a type with value $\bar{\val}$, 
which is also upper bounded by $\int_0^{\bar{\val}}(1-\alloc(z)) \dd z$.
\end{proof}

\lemCorrelateSupply*
\begin{proof}
We first show that given any sale probability constraint $q$, the ex-ante optimal mechanism has menu complexity of $3$. 
Similar to the proof of \cref{thm:correlate menu}, 
for any ex-ante feasible mechanism $\mech$, 
let $\val_0$ be the minimum value of the agent that participates the auction with utility $u_0$, 
and let $\alloc_0 = \frac{u_0}{v_0}$.
Consider the following revenue maximization problem
without concerns for participation costs. 
The seller can only sell the item to agents with value above $\val_0$ subject to the incentive constraint
and the sale probability constraint $q$. 
In addition, the minimum allocation of the agent is~$\alloc_0$, 
and the utility of the agents with value $\val_0$ is $u_0$.
It is easy to verify that by \citet{alaei2013simple}, 
the revenue optimal mechanism $\mech'$ in this setting has at most two steps. 
Since the minimum allocation is at least $\alloc_0$, 
this corresponds to a mechanism with menu size $3$.
Moreover, one of the menu entries is $(x_0, 0)$. 
Note that this menu entry does not contribute to the expected revenue. 
Again similar to the proof of \cref{thm:correlate menu},
both $\mech$ and $\mech'$ have the same revenue across the two settings, 
and hence in the original setting, $\Rev(\mech)\leq \Rev(\mech')$.

Now we focus our attention on the original setting. 
In mechanism $\mech'$, for any menu entry $(\alloc,\pay)$
with $p>0$, 
let $q_x$ be the probability the agent chooses this menu entry in the ex-ante optimal mechanism. 
By definition we have $q_x < q$.
Thus by posting price $\pay$ to the agent, 
with probability $q_p \geq q_x$, 
the agent will accept the price and the expected revenue is at least $p \cdot q_p \geq p \cdot q_x$. 
If $q_p \leq q$, then posting price $p$ is feasible and generates higher revenue than the contribution of the menu entry $(\alloc,\pay)$ in the ex-ante optimal mechanism. 
If $q_p > q$, there exists a higher price $\hat{p} \geq p$ such that the item is sold with probability exactly $q$. 
In this case, the revenue from posting price $\hat{p}$
is $\hat{p}q \geq pq \geq p\cdot q_x$.
Therefore, by posting the optimal price such that the item is sold with probability at most~$q$, 
the revenue is a $2$-approximation to the ex-ante optimal.
\end{proof}

\section{Negative Participation Costs}
\label{apx:negative outside}
In this section, we show that posted pricing is approximately optimal when the participation cost can take negative value with positive probability.
Note that in this case, by \cref{lem:p0=0}, 
the parameter in the payment function satisfies $\pay_0\geq 0$. 

\begin{proposition}\label{lem:dmr supply negative}
For the single-buyer setting, 
if the conditional value distribution $\jointF_c$ has
identical and bounded support,
and has decreasing marginal revenue for any participation cost~$\cost$,
there exists a deterministic mechanism with at least half of the optimal revenue.
\end{proposition}
\begin{proof}
Let $H<\infty$ be the maximum value of the buyer. 
For any allocation rule $\alloc$ 
and associated payment rule with parameter $p_0\geq 0$,
let $q$ be the probability the agent participates the auction
given allocation and payment $\alloc$ and $\pay$.
Let $\bar{\val} = \sup_{\val\leq H} \{\alloc(\val) < 1\}$,
$\mu = \int_{0}^{\bar{\val}} (1-\alloc(z)) \dd z$,
and let 
\begin{align*}
\hat{\alloc}(\val) = \begin{cases}
1 & \val \geq \mu\\
0 & \val < \mu.
\end{cases}
\end{align*}
For any participation cost $\cost$, 
it is easy to verify that 
\begin{align*}
\hat{\alloc}_{\cost}(\val) = \begin{cases}
1 & \val \geq \max\{\val_{\cost}(\hat{\alloc}),\mu\}\\
0 & \val < \max\{\val_{\cost}(\hat{\alloc}),\mu\},
\end{cases}
\end{align*}
where $\val_{\cost}(\hat{\alloc}) = 0$
if $\cost \leq - p_0$
and $\val_{\cost}(\hat{\alloc}) = 
\bar{\val} - \mu + \cost + p_0$
if $\cost > - p_0$.
Moreover, $\int_0^{\vupper} \alloc(z) \dd z = \int_0^{\vupper} \hat{\alloc}(z) \dd z$ 
and $\int_0^{\val} \alloc(z) \dd z \geq \int_0^{\val} \hat{\alloc}(z) \dd z$ for any $\val\geq 0$. 
By Inequality \eqref{eq:rev increase}, 
for any $\cost\geq 0$,
we have 
$\rev(\alloc,p_0; \cost) \leq \rev(\hat{\alloc},p_0; \cost)$. 
Next we consider the case that $\cost < 0$. 
For the case that $-p_0<\cost < 0$, 
the revenue of mechanism with parameter $p_0$ is
\begin{align*}
\rev(\alloc, p_0; \cost) 
&= 
\int_{\val_{\alloc}(\cost)}^{\vupper} \jointf_c(\val)\alloc_{\cost}(\val)\virtual_c(\val) \dd\val
- (1-\jointF_c(\val_{\alloc}(\cost)))\cdot \cost \\
&\leq \int_{\val_{\hat{\alloc}}(\cost)}^{\vupper} \jointf_c(\val)\hat{\alloc}_{\cost}(\val)\virtual_c(\val) \dd\val
- (1-\jointF_c(\val_{\alloc}(\cost)))\cdot \cost \\
&\leq \int_{\val_{\hat{\alloc}}(\cost)}^{\vupper} \jointf_c(\val)\hat{\alloc}_{\cost}(\val)(\virtual_c(\val) - \cost) \dd\val
- (1-\jointF_c(\val_{\alloc}(\cost)))\cdot \cost \\
&= \rev(\hat{\alloc}, p_0; \cost) - (1-\jointF_c(\val_{\alloc}(\cost)))\cdot \cost .
\end{align*}
The first inequality holds by applying Inequality \eqref{eq:rev increase}, 
and the second inequality holds since $c < 0$.
Finally we consider the case $c \leq -p_0$.
Here the buyer will always participate the auction for all values, $\val_{\alloc}(\cost) = \val_{\hat{\alloc}}(\cost) = 0$,
and 
\begin{align*}
\rev(\alloc, p_0; \cost) 
&= 
\int_{0}^{\vupper} \jointf_c(\val)\alloc_{\cost}(\val)\virtual_c(\val) \dd\val
+ p_0 \\
&\leq \int_{0}^{\vupper} \jointf_c(\val)\hat{\alloc}_{\cost}(\val)\virtual_c(\val) \dd\val
+ p_0 
= \rev(\hat{\alloc}, p_0; \cost),
\end{align*}
and the inequality holds again by applying Inequality \eqref{eq:rev increase}.
Combining three cases and taking expectation over $\cost$, 
we have 
\begin{align*}
\rev(\alloc, p_0) 
&= \expect[\cost\sim G]{\rev(\alloc, p_0; \cost)}\\
&\leq \expect[\cost\sim G]{\rev(\hat{\alloc}, p_0; \cost)}
- \expect[\cost\sim G]{(1-\jointF_c(\val_{\alloc}(\cost)))\cdot \cost \cdot \indicate{-p_0<\cost < 0}}\\
&\leq \rev(\hat{\alloc}, p_0)
+ \expect[\cost\sim G]{(1-\jointF_c(\val_{\alloc}(\cost)))\cdot p_0}.
\end{align*}
Note that the term of $\expect[\cost\sim G]{(1-\jointF_c(\val_{\alloc}(\cost)))\cdot p_0} = p_0\cdot q$.
Moreover, there exists a mechanism that
charges price $p_0$ for the item, 
and buyer participates in it with probability at least $q$.
The revenue of this posted price mechanism is at least $p_0\cdot q$.
Thus the revenue of allocation $\alloc$ with with parameter $p_0\geq 0$
is upper bounded by twice of the optimal revenue of posted pricing.
\end{proof}

\section{Hoeffding's inequality}
\begin{lemma}[Hoeffding's inequality]\label{lem:Hoeffding}
For any Bernoulli random variable with probability $p$ for value~$1$,
let $H(n)$ be the number of $1$ values given $n$ trials. 
For any $\epsilon > 0$, we have 
\begin{align*}
\Pr\left[ \left|\frac{1}{n}H(n) - p\right| \geq \epsilon \right] \leq 2\exp(-2\epsilon^2 n).
\end{align*}
\end{lemma}

%% file: fig/DMR-x.tex
\begin{tikzpicture}[scale = 0.55]

\draw (-0.2,0) -- (12.5, 0);
\draw (0, -0.2) -- (0, 6);

\begin{scope}[thick]
\draw (0, 0) -- (3.9, 0);
\draw (3.9, 5.5) -- (12, 5.5);
\end{scope}
\draw [dotted] (3.9, 0) -- (3.9, 5.5);

\draw [dashed] plot [smooth, tension=0.6] coordinates { (0,0) (1.5,2) (5, 3) };
\draw [dashed] plot [smooth, tension=0.6] coordinates {(5, 3) (8,5) (12,5.5)};

\draw [dotted] (6.3, 0) -- (6.3, 5.5);
\draw (6.3, -0.7) node {$v_c(\hat{x})$};

\draw (0, -0.7) node {$0$};
\draw (5.1, 2.4) node {$c$};

\draw (12, -0.7) node {$\bar{v}$};
\draw [dotted] (12, 0) -- (12, 5.5);
\draw (3.9, -0.7) node {$\mu$};

\fill[pattern=north west lines, pattern color=gray!40!white] (3.9, 0) -- (6.3, 0) 
-- (6.3, 5.5) -- (3.9, 5.5);

\draw (8, 6.1) node {$\hat{x}(v)$};
\draw (8, 4.3) node {$x(v)$};

\draw (-0.4, 5.5) node {$1$};
\draw [dotted] (0, 5.5) -- (3.9, 5.5);

\end{tikzpicture}

%% file: fig/MHR-x.tex
\begin{tikzpicture}[scale = 0.55]

\draw (-0.2,0) -- (12.5, 0);
\draw (0, -0.2) -- (0, 6);

\draw [name path = C, dotted, white] (4, 5.5) -- (8, 5.5);
\begin{scope}[very thick]
\draw [dashed] plot [smooth, tension=0.6] coordinates { (0,0) (1.5,2) (8, 4) };
\draw [dashed] plot [smooth, tension=0.6] coordinates {(8, 4) (9,5) (12,5.5)};
\end{scope}

\begin{scope}[very thick, blue]
\draw [dashed] (4, 5.5) -- (12, 5.5);
\draw [dashed] (0, 0) -- (4, 0);
\end{scope}
\draw [name path = A, dotted] (4, 0) -- (4, 5.5);

\draw [name path = B, line width=0.18mm, red] plot [smooth, tension=0.6] coordinates { (0,0) (1.5,2) (8, 4) };
\draw [line width=0.18mm, red] (8, 5.5) -- (12, 5.5);
\draw [name path = D, dotted, gray!50!white] (8, 4) -- (4, 5.5);

\draw [dotted] (8, 0) -- (8, 5.5);

\tikzfillbetween[of=A and B, every even segment/.style={pattern=crosshatch}]{pattern = north west lines, pattern color=gray!50!white};
\tikzfillbetween[of=C and D, every even segment/.style={pattern=crosshatch}]{pattern = north west lines, pattern color=gray!50!white};

\draw (0, -0.7) node {$0$};

\draw (3.9, -0.7) node {$\mu$};

\draw (8, -0.7) node {$v'=\bar{v}$};

\draw (-0.4, 5.5) node {$1$};
\draw [dotted] (0, 5.5) -- (4, 5.5);

\end{tikzpicture}

%% file: fig/concave.tex
\begin{tikzpicture}[scale = 0.55]

\draw (-0.2,0) -- (10.5, 0);
\draw (0, -0.2) -- (0, 6);

\begin{scope}[thick]
\draw [dashed] plot [smooth, tension=0.5] coordinates {(0, 0) (2, 2) (8, 2.5)}; 
\end{scope}

\draw plot [smooth, tension=0.6] coordinates {(0, 0) (2.5, 0.7) (6, 3) (8,5.5)}; 

\draw [red] (0, 0) -- (6.6, 3);
\draw [red] (6.6, 3) -- (8, 5.5);

\draw [dotted] (5.2, 0) -- (5.2, 2.4);
\draw (5.2, -0.7) node {$v_0$};

\draw [dotted] (0, 2.4) -- (5.2, 2.4);
\draw (-0.8, 2.4) node {$u_0$};

\draw (0, -0.7) node {$0$};

\end{tikzpicture}